%% file: main.tex
\documentclass[12pt]{article}
\usepackage{amsthm,amsfonts,amsmath, amscd}
\usepackage{color}

\swapnumbers
                              
\theoremstyle{plain}
\newtheorem{theorem}{THEOREM}[section]
\newtheorem{lemma}[theorem]{LEMMA}
\newtheorem{corollary}[theorem]{COROLLARY}
\newtheorem{proposition}[theorem]{PROPOSITION}

\theoremstyle{definition}
\newtheorem{definition}[theorem]{DEFINITION}
\newtheorem{remark}[theorem]{Remark}
\theoremstyle{remark}

\usepackage{type1cm}         

\usepackage{makeidx}         
\usepackage{graphicx}        
\usepackage{multicol}        
\usepackage[bottom]{footmisc}

\usepackage{newtxtext}       %
\usepackage[varvw]{newtxmath}       

                       \usepackage{amsfonts,amsmath, amscd,dsfont}
\usepackage{mathrsfs}

\newcommand{\R}{{\mathord{\mathbb R}}}
\newcommand{\one}{{\mathds1}}
\newcommand{\C}{{\mathord{\mathbb C}}}

\newcommand{\N}{{\mathord{\mathbb N}}}

\newcommand{\cA}{{\mathord{\cal A}}}

\newcommand{\tr}{{\rm Tr}}
\numberwithin{equation}{section}
\newcommand{\K}{{\mathcal K}}
\newcommand{\un}{{\rm 1\kern -2.5pt l}}

 \newcommand{\Dens}{{\mathfrak S}}
 \newcommand{\cG}{\mathscr{G}}

\newcommand{\cJ}{\mathcal{J}}

\newcommand{\cM}{\mathscr{M}}
\newcommand{\cH}{\mathcal{H}}
\newcommand{\cK}{\mathcal{K}}
\newcommand{\cD}{\mathcal{D}}

\newcommand{\cL}{\mathscr{L}}

\newcommand{\cB}{\mathcal{B}}

\newcommand{\sH}{\mathscr{H}}
\newcommand{\cP}{\mathscr{P}}
\newcommand{\cQ}{\mathscr{Q}}

\newcommand{\ip}[1]{\langle {#1}\rangle}

\newcommand{\cF}{\mathcal{F}}

\DeclareMathOperator{\Ent}{Ent}
\DeclareMathOperator{\dive}{div}
\newcommand{\rmD}{\mathrm{d}}

\newcommand{\eps}{\varepsilon}
\renewcommand{\phi}{\varphi}

\newcommand{\dd}{\; \mathrm{d}}

\newcommand{\fH}{{\mathfrak{H}}}

\begin{document}

\author{Eric Carlen\\ Department of Mathematics, Hill Center,\\ Rutgers University,
110 Frelinghuysen Road
Piscataway NJ 08854-8019 USA}
\title{Dynamics and Quantum Optimal Transport:
Three lectures on quantum entropy and quantum Markov semigroups}

\date{June 19, 2023}

\maketitle

\begin{abstract}
\footnotesize{ 

This document presents the contents of three lectures delivered by the author at the Erd\H{o}s Center School ``Optimal Transport on Quantum Structures'', Septemer 19-23, 2022 in Budapest, Hungary. It presents a fairly self contained account of an active topic of current research, and this account should be accessible to most graduate students, as befits lectures for a school. The main results are known, but there a number of new proofs and some  new results. 
}
\end{abstract}



\tableofcontents


\include{chapterC1}

\include{chapterC2}

\include{chapterC3}

\input{references}
\end{document}

%% file: chapterC1.tex
%
%
%
\section{Lecture One: Quantum Entropy Inequalities}

\abstract{\footnotesize Inequalities for quantum entropy and Fisher information play an important role in the theory of quantum Markov semigroups, and the first lecture is devoted to a comprehensive treatment of results that we shall use, including complete proofs. One of the key results is the Data Processing Inequality, which is one of the many equivalent forms of the Strong Subadditivity of Quantum Entropy, first proved by Lieb and Ruskai in 1973.
This inequality says that the quantum relative entropy between any two quantum states  decreases under any {\em quantum operation}; that is, completely positive and trace preserving map.}

\subsection{Introduction}

   The following notation will be used throughout these lectures. $M_n(\C)$ denotes the  $n\times n$ complex matrices, 
$M_n^+(\C)$ the positive semidefinite  $n\times n$ complex matrices, and 
$M_n^{++}(\C)$ the  positive definite  $n\times n$ complex matrices.   Sometimes it is more convenient to refer instead to bounded operators on a 
Hilbert space $\cH$, in which case we use the corresponding notations $\cB(\cH)$, $\cB^+(\cH)$ and $\cB^{++}(\cH)$.  To keep the proofs simple, we shall always assume that $\cH$ is finite dimensional. 
This is the heart of the matter and is the case most relevant to applications in quantum information theory. However, many of the results do extend easily to infinite dimensions. 
In particular, for the tools necessary to generalize  the content of the first lecture to a fairly general von Neumann algebra setting, see Araki's paper \cite{A75}. 
Araki's paper also makes an explicit  connection with the matrix algebra case treated here,

A {\em density matrix} $\rho$ is an element of $M_n^+(\C)$ with unit trace; i.e., $\tr[\rho] =1$. 
The set of all {\em invertible density matrices} is denoted $\mathfrak{S}$, and its closure, the set of all density matrices is $\overline{{\mathfrak S}}$. For a quantum observable represented by a self-adjoint operator $A$, the expected value of the measurement of $A$ with the system in the state $\rho$ is $\tr[\rho A]$.  For a density matrix $\rho$, the {\em von Neumann entropy} $S(\rho)$
is the quantity
$$
S(\rho) := -\tr[\rho \log \rho]\ .
$$
Given two density matrices $\rho$ and $\sigma$, the 
{\em Umegaki relative entropy of $\rho$ with respect to $\sigma$} \cite{U62} is the quantity
\begin{equation}\label{umegaki}
D(\rho ||\sigma)  := \tr[\rho(\log \rho - \log \sigma)]\ .
\end{equation}
This is one of several generalizations of the classical formula for the relative entropy of two  strictly positive probability distribution $p$ and $q$ on $\{1,\dots,n\}$, namely
$$
D(p||q) = \sum_{j=1}^n p_j \log\left(\frac{p_j}{q_j}\right)\ .
$$
There are various natural interpretations of  ``logarithm of $\rho$ divided by $\sigma$'' for two density matrices $\rho,\sigma\in \Dens$.  One, the choice made by Umegaki, is $\log\rho - \log \sigma$. 
Another natural choice is  
${\displaystyle  \log\left(\rho^{1/2}\sigma^{-1}\rho^{1/2}\right)}$.
Making this other choice yields the {\em Belavkin-Stasewski relative entropy} \cite{BS82}
$$
D_{BS}(\rho ||\sigma) := \tr\left[ \rho  \log\left(\rho^{1/2}\sigma^{-1}\rho^{1/2}\right)\right]\ .
$$
There are many other generalizations of the classical formula, some of which lead to meaningful relative entropy functionals that share many key 
properties of the classical relative entropy functional, and others that do not. 
For more on this, see \cite{Do86,CL19}. 

Having pointed out the there are other interesting  functionals that merit the name ``quantum relative entropy'', the rest of these lectures  focus on the Umegaki relative entropy defined in \eqref{umegaki}, and when we refer to quantum relative entropy in what follows, we always mean the Umegaki relative entropy, which has a direct physical meaning, as we now explain.

Loosely speaking, the number of experiments required to distinguish the state $\rho$ from $\sigma$ with an error probability of size 
$\epsilon$  is $-\log \epsilon/D(\rho ||\sigma)$, in a complete analogy with ``hypothesis testing'' in classical mathematical statistics.   See \cite{HP91,JP06,ON00,OP93} for more information and  the introduction to \cite{CFL18} for a brief resume. 
In this sense, the relative entropy is a physically relevant measure of the ``divergence'' of $\rho$ from $\sigma$.  While the relative entropy is not a metric -- $D(\rho ||\sigma)$ is not even symmetric in $\rho$ and $\sigma$ -- it does dominate the square of a natural metric, as expressed by the quantum Pinsker inequality
\begin{equation}\label{Pinsker}
D(\rho||\sigma) \geq \frac12 \tr[|\rho - \sigma|]^2\ .
\end{equation}

The evolution of quantum states, i.e., density matrices, in non-equilibrium quantum statistical mechanics is  given by {\em quantum Markov 
Semigroups}, which are semigroups of {\em quantum operations}, as we now explain.

Let $\Phi:M_m(\C) \to M_n(\C)$ be a linear transformation. $\Phi$ is {\em positive} when $\Phi(A) \geq 0$ for all $A \geq 0$. $\Phi$ is $2$-positive if the block matrix
 $\left[\begin{array}{cc} \Phi(A) & \Phi(B) \\\Phi(C )& \Phi(D)\end{array}\right] \geq 0$ whenever  $\left[\begin{array}{cc} A & B \\C & D\end{array}\right] \geq 0$.  For each integer $k > 2$, the condition of $k$-positivity is defined in the analogous manner, and $\Phi$ is {\em completely} positive if it is $k$ positive for all $k$ \cite{S55}.   For example, it is easy to see that for any $m\times n$ matrix $V$,
 the map $\Phi: X \mapsto V^* XV$ is completely positive, and it follows that for any set $\{V_1,\dots,V_\ell\}$ of $m\times n$ matrices,
\begin{equation}\label{lieb6}
 \Phi: X \mapsto \sum_{j=1}^\ell V_j^* X V_j\ .
\end{equation}
 is completely positive. By a theorem of Kraus \cite{K71} and Choi \cite{Choi72,Choi75}, every completely positive map $\Phi$ from $M_m(\C)$ to $M_n(\C)$ has this form, and in fact, one can always find such a representation with $\ell \leq mn$.  The map $\Phi$ is {\em unital} if it takes the identity to the identity; i.e., $\Phi(I) = I$.  The map $\Phi$ is {\em trace preserving} if $\tr[\Phi(X)] = \tr[X]$ for all $X$.  For each $n$, we equip $M_n(\C)$ with the Hilbert-Schmidt inner product $\langle X,Y\rangle := \tr[X^*Y]$, making it a Hilbert space. For any linear transformation $\Phi:M_m(\C) \to M_n(\C)$, we write $\Phi^\dagger$ to denote the adjoint with respect to the Hilbert-Schmidt inner product. It is easy to see that if $\Phi$ is given by \eqref{lieb6}, then 
 \begin{equation*}
 \Phi^\dagger : X \mapsto \sum_{j=1}^\ell V_j X V_j^*\ .
\end{equation*}
In any case, whether $\Phi$ is completely positive or not, $\Phi$ is unital if and only if $\Phi^\dagger$ is trace preserving. 

Unital completely positive maps $\Phi$ and their completely positive trace preserving duals $\Phi^\dagger$ play a fundamental role in quantum information theory and the quantum theory of open systems \cite{D76}.   Since the state of a finite dimensional quantum system may be identified with a  density matrix 
$\rho\in \overline{{\mathfrak S}}$, a linear map $\Phi^\dagger$ that updates the state of a quantum system must at the very least preserve positivity and the trace. But this is not the only requirement. The condition of complete positivity arises through the consideration of a system relative to its environment, and hence it is completely positive unital maps  that update the observables in the time evolution, and completely positive trace preserving maps that update the states under a partial measurement of a system \cite{K71}.   Completely positive trace preserving maps maps $\Phi^\dagger$  are known  as {\em quantum operations}, and they give the {\em Sch\"odinger picture} of the evolution of a system, with the states being updated. In the dual {\em Heisenberg picture}, the state is unchanged, but the observables are updated by $\Phi$. 
In the literature on quantum Markov semigroups, it is fairly 
standard to use $\Phi$  to refer to the Heisenberg picture and $\Phi^\dagger$ to refer to the Schr\"odinger picture.  We shall generally adhere to this convention here.

 The following example is  fundamental:  Let $m,n\in \N$ and  identify $\C^{mn}$ with the $m$-fold direct sum of $m$ copies of $\C^n$. We may then think of matrices in $M_{mn}(\C)$  as $m\times m$ block matrices each of whose entries is in $M_n(\C)$.  Then, suppressing the $n$ dependence, for each $m$ we define
\begin{equation*}
\Psi_m(X) = \left[\begin{array}{ccc} X & \phantom{X} &  \phantom{X}\\   \phantom{X} & \ddots &  \phantom{X} \\  \phantom{X} &  \phantom{X} & X\end{array}\right]\ ,
\end{equation*}
where the matrix on the right is the $m\times m$ block diagonal  matrix each of whose diagonal entries is $X$.  This is evidently completely positive and  unital.   Its adjoint, $\Psi_m^\dagger$, is therefore trace preserving and completely positive. It is easy to see that 
\begin{equation*}
\Psi^\dagger_m\left(\left[\begin{array}{ccc} 
X_{1,1} & \cdots & X_{1,m}\\ \vdots & \ddots &\vdots\\
X_{m,1} & \cdots & X_{m,m}\end{array}\right]\right) = \sum_{j=1}^m X_{j,j}\ ,
\end{equation*} and this operation is known as the {\em partial trace}.   
Some of the theorems that follow involve negative powers. When $\Phi:M_n(\C) \to M_m(\C)$ is unital and completely positive, $\Phi(A)\in M_m^{++}(\C)$ whenever $A\in M_n^{++}(\C)$.  However, the image of $M_m^{++}(\C)$ under the completely positive trace preserving map $\Phi^\dagger$  may lie entirely in $M_n^{+}(\C) \backslash M_n^{++}(\C)$.  For example, take $n=2m$ and think of elements of $M_n(\C)$ as $2\times 2$ block matrices with entries in $M_m(\C)$. Then  define
$$
\Phi\left(\left[\begin{array}{cc} A & B\\ C & D\end{array}\right]\right) = A  \qquad{\rm so \ that}\qquad \Phi^\dagger(A) =  \left[\begin{array}{cc} A & 0\\ 0 & 0\end{array}\right]\ .
$$

A fundamental property of the quantum relative entropy is that it decreases under any quantum operation. This is the {\em Data Processing Inequality} (DPI) proved by Lindblad \cite{Lind75} building on earlier work \cite{Lind73,Lind74}, and it says that for any two density matrices $\rho$ and $\sigma$, and any quantum operation $\Phi^\dagger$,
\begin{equation}\label{DPI}  
D(\Phi^\dagger\rho || \Phi^\dagger \sigma) \leq D(\rho ||\sigma)\ .
\end{equation}
Roughly speaking, performing any quantum operation on $\rho$ and $\sigma$ can only make then harder to distinguish. This inequality is one of the cornerstones of quantum information theory. 

We now turn to semigroups of quantum operations, and make our first connection relating them and quantum entropy. 

\begin{definition} A {\em Quantum Markov Semigroup} (QMS) is a semigroup $\{\cP_t\}_{t\geq 0}$ of completely positive unital maps on $M_n(\C)$. Its dual 
$\{\cP_t^\dagger\}_{t\geq 0}$ is a semigroup of completely positive trace preserving maps. Such semigroups are called {\em Quantum Dynamical Semigroups}. 
\end{definition} 

For any QMS $\{\cP_t\}_{t\geq 0}$, and any $t>0$ since $\cP_t\one = \one$,  $1$ is an eigenvalue of each $\cP_t$, and hence each $\cP_t^\dagger$. 
Therefore, there will be at least one invariant density matrix $\sigma$ for the dynamical semigroup $\{\cP_t\}_{t\geq 0}$; i.e., a density matrix
$\sigma$ such that $\cP_t\sigma = \sigma$ for all $t\geq 0$. We will generally be interested in the case in which there is exactly one such state $\sigma$; i.e., the case in which the multiplicity of the eigenvalue $1$ of $\cP_t^\dagger$ is $1$.  In this case, the QMS is said to be {\em ergodic}.

Consider such a quantum dynamical semigroup $\{\cP_t^\dagger\}_{t\geq 0}$ with invariant state $\sigma$. Then by the invariance of $\sigma$ and the DPI, for any density matrix $\rho$.
$$
D(\cP_t^\dagger \rho || \sigma) = D(\cP_t^\dagger \rho || \cP_t^\dagger \sigma)  \leq D(\rho ||\sigma)\ .
$$
That is, in complete generality, the function
$$
t \mapsto  D(\cP_t^\dagger \rho || \sigma)
$$
is monotone decreasing. The following list of questions motivate the discussion that follows:

\bigskip
\noindent{\it (1)} Under what circumstances do we have $\lim_{t\to\infty}D(\cP_t^\dagger \rho || \sigma) =0$, and when this is the case, when is the convergence exponential and at what  exponential rate? (Note that by the Pinsker Inequality \eqref{Pinsker}, when $D(\cP_t^\dagger \rho || \sigma) $ decreases at an exponential rate, then $\tr[|\cP_t^\dagger \rho - \sigma|]$ decreases at one-half this exponential rate.) 

\bigskip
\noindent{\it (2)}  Under what circumstances can we write the evolution $\rho \mapsto \cP_t^\dagger \rho$ as gradient for for the relative entropy with respect to some metric on the space $\mathfrak{S}$ of invertible density matrices?

\bigskip
\noindent{\it (3)}  When it is possible to write the evolution $\rho \mapsto \cP_t^\dagger \rho$ as gradient for the relative entropy with respect to some metric on the space $\mathfrak{S}$ of invertible density matrices,
under what circumstances can we explicitly find such a metric  on $\mathfrak{S}$ and relate the rate at which $\lim_{t\to\infty}D(\cP_t^\dagger \rho || \sigma) =0$ to geometric properties of the metric? 

\bigskip

Throughout this investigation, a number of inequalities pertaining to quantum entropy play a fundamental role.  We begin the discussion proper with a self-contained presentation of these, with complete proofs.

\subsection{Convexity theorems and monotonicity theorems}

The DPI \eqref{DPI} is a prototypical monotonicty theorem: It says that the relative entropy is monotone decreasing under the application of any completely positive trace preserving map. This is closely connected with the fact that 
$$
(\rho,\sigma) \mapsto D(\rho ||\sigma)
$$
is convex; this is the {\em joint convexity of the relative entropy}.  To see that this is a simple consequence of the DPI, consider the map  $\Phi: M_n(\C) \to M_{2n}(\C)$ be defined by 
\begin{equation*}
\Phi(H) = \left[\begin{array}{cc} H & 0\\ 0 & H\end{array}\right] \quad{\rm and\ hence}\quad 
\Phi^\dagger\left(\left[\begin{array}{cc} Y_{1} & 0\\0 & Y_{2}\end{array}\right]\right) = Y_{1}+Y_{2}\ . 
\end{equation*}
 These are evidently completely positive and $\Phi^\dagger$ is evidently trace preserving. Now consider two pairs of density matrices $\{\rho_1,\rho_2\}$ and $\{\sigma_1,\sigma_2\}$ and define
 $$
 \rho := \frac12\left[\begin{array}{cc} \rho_1 & 0 \\0 & \rho_2\end{array}\right] \qquad{\rm and}\qquad 
 \sigma := \frac12\left[\begin{array}{cc} \sigma_1 & 0 \\0 & \sigma_2\end{array}\right] \ .
 $$
 Then the DPI says
 $$D(\rho||\sigma) \geq D\left(\frac{\rho_1+\rho_2}{2}\big|\big|\frac{\sigma_1+\sigma_2}{2}\right)\ ,$$
 but clearly
 $$D(\rho||\sigma) = \frac12 D(\rho_1||\sigma_1) + \frac12 D(\rho_2||\sigma_2)\ .$$
 Therefore
 $$\frac12 D(\rho_1||\sigma_1) + \frac12 D(\rho_2||\sigma_2) \geq  D\left(\frac{\rho_1+\rho_2}{2}\big|\big|\frac{\sigma_1+\sigma_2}{2}\right)\ .$$
 By continuity, midpoint convexity implies convexity in general, and this proves that the joint convexity is a consequence of the DPI. 
 
 It turns out that the DPI is also a consequence of the joint convexity of the relative entropy. Historically, this is the way the DPI was first proved \cite{Lind75}: 
 The joint convexity of the relative entropy is a simple consequence of the Lieb Concavity Theorem \cite{L73}, and from here, Lindblad used the 
 structure theory of completely positive maps given by the Stinespring Dilation Theorem to prove the DPI. See \cite{C22} for more information.

Many well-known concavity or convexity theorems for matrices have an equivalent expression as a monotonicity theorem, just as we 
have explained in regard to the joint convexity of the relative entropy and the DPI.  It turns out that the monotonicity formulations have a number of advantages:

\medskip
\noindent{\it (1)} It is often simpler to give a direct proof of the monotonicity theorem than a direct proof of the convexity/concavity theorem but the former always implies the latter via consideration of the partial trace.

\medskip
\noindent{\it (2)} The monotonicity theorem often has direct physical significance, as in the case of the DPI.

\medskip
\noindent{\it (3)} The monotonicity theorem often holds for a wider class of positive maps than quantum operations. A case in point is the DPI which 
was recently shown by M\"uller-Hermes and Reeb \cite{MHR17}  to be valid for all positive trace preserving maps.

The  fundamental monotonicity  inequalities that we will prove here are the monotonicity versions of the first three convexity theorems in a fundamental 1973 paper of Lieb \cite{L73}:

\begin{theorem}\label{L1M}   For all completely positive unital maps $\Phi: M_n(\C)\to M_m(\C)$,  all
$0 \leq t \leq 1$, all $m,n\in \N$, all $X,Y \in M_m^+(\C)$, all $K\in M_n(\C)$,  
 \begin{equation}\label{lieb21}
 \tr[\Phi(K^*)Y^{1-t} \Phi(K) X^t]  \leq  \tr[K^*\Phi^\dagger(Y)^{1-t} K \Phi^\dagger(X)^t] \ .
\end{equation}
\end{theorem}

\begin{theorem}\label{L2M}   For all completely positive unital maps $\Phi: M_n(\C)\to M_m(\C)$,  
 all $0 \leq t \leq 1$, all $m,n\in \N$, all 
$X,Y \in M_m^{++}(\C)$ and all $K\in M_m(\C)$ such that $\Phi^\dagger(X),\Phi^\dagger(Y) \in M_n^{++}(\C)$,  
 \begin{equation}\label{lieb22} 
 \tr[ \Phi^\dagger(K^*)\Phi^\dagger(Y)^{t-1}  \Phi^\dagger(K)\Phi^\dagger(X)^{-t}]  \leq   \tr[ K^*Y^{t-1}  KX^{-t}]\ .
\end{equation} 
\end{theorem}

\begin{theorem}\label{L3M}   For all completely positive unital maps $\Phi: M_n(\C)\to M_m(\C)$, all $0 \leq t \leq 1$, all $m,n\in \N$, all $X,Y \in M_m^{++}(\C)$, all $K\in M_m(\C)$, 
\begin{equation}\label{lieb23} 
 \tr \left[   \int_0^\infty  \Phi^\dagger(K^*)\frac{1}{sI + \Phi^\dagger(Y)} \Phi^\dagger(K) \frac{1}{sI + \Phi^\dagger(X)}{\rm d}s \right]   
  \leq   \tr \left[   \int_0^\infty K^*\frac{1}{sI + Y} K \frac{1}{sI + X}{\rm d}s \right]\ .
\end{equation}
\end{theorem}

Theorem~\ref{L1M} was first proved by Uhlmann \cite{Uh77} in a  more general -- and  definitive --  form.  He did not need to assume that $\Phi$ is completely positive or unital, but only that it satisfies the inequality
$$\Phi(K^*K) \geq \Phi(K)^*\Phi(K)\ .$$
This is  known the {\em Schwarz inequality} in operator algebras, and as we will explain below, it holds for all unital completely positive maps \cite{Choi74}, but also for a wider class of positive maps.
 In the more restricted form stated here, all that we will need, Theorem~\ref{L1M} is equivalent to the  celebrated Lieb Concavity Theorem, Theorem~\ref{L1} below, as we shall explain. 

First, let us see how Theorem~\ref{L1M} yields the DPI. Take $K =1$. Since $\Phi$ is unital, \eqref{lieb23} reduces to 
\begin{equation*}
 \tr[Y^{1-t} X^t]  \leq  \tr[\Phi^\dagger(Y)^{1-t}  \Phi^\dagger(X)^t] \ ,
\end{equation*}
Since $\Phi^\dagger$ is trace preserving, $ \tr[Y] =\tr[\Phi^\dagger(Y)] $, and hence for all $t>0$,
$$
\frac{ \tr[Y^{1-t} X^t]  - \tr[Y]}{t}  \leq  \frac{\tr[\Phi^\dagger(Y)^{1-t}  \Phi^\dagger(X)^t]  - \tr[\Phi(Y)]}{t}\ .
$$
Taking the limit $t\downarrow 0$, we obtain
$$
\tr[ Y(\log X - \log Y)] \leq \tr[\Phi^\dagger(Y)(\log \Phi^\dagger(X) - \log \Phi^\dagger(Y)] \ ,
$$
and this is the same as $D(\Phi^\dagger(X)|| \Phi^\dagger(Y)) \leq D(X||Y)$.   

We next explain how Theorem~\ref{L1M} yields the Lieb Concavity Theorem:

\begin{theorem}[Lieb, 1973]\label{L1}
For $0 \leq t \leq 1$, and any fixed $K\in M_n(\C)$, the function
\begin{equation}\label{lieb1}
(X,Y) \mapsto \tr[K^*Y^{1-t} K X^t]
\end{equation}
is jointly concave on $M_n^+(\C)\times M_n^+(\C)$. 
\end{theorem}

\begin{proof} Let $X_1$, $X_2$, $Y_1$ and $Y_2$ be positive definite in $M_m(\C)$, and define
$$
X = \left[\begin{array}{cc} X_1 & 0 \\ 0 & X_2\end{array}\right]\qquad{\rm and}\qquad Y = \left[\begin{array}{cc} Y_1 & 0 \\ 0 & Y_2\end{array}\right]\ .
$$
As before, define $\Phi: M_m(\C)\to M_{2m}(\C)$ by $\Phi(K) = \left[\begin{array}{cc} K & 0\\ 0 & K\end{array}\right]$ for all  $K\in M_n(\C)$, 
and note that $\Psi$  is completely positive and unital. Then
$$
 Y^{1-t}\Phi(K)X^t  =   \left[\begin{array}{cc} Y_1^{1-t}K X_1^t & 0 \\ 0 & Y_2^{1-t}KX_2^t\end{array}\right]
$$
so that $\tr[\Phi(K^*) Y^{1-t}\Phi(K)X^t ] = \tr[K^*Y_1^{1-t}K X_1^t] + \tr[K^*Y_2^{1-t}KX_2^t]$.  On the other hand,
$$\tr[ K^*\Phi^\dagger(Y)^{1-t}K \Phi^\dagger(X)^t ] = \tr[K^*(Y_1+Y_2)^{1-t} K (X_1+X_2)^t]\ .$$
Then by \eqref{lieb21}
\begin{equation*}
 \tr[K^*Y_1^{1-t}K X_1^t] + \tr[K^*Y_2^{1-t}KX_2^t]  \leq   \tr[K^*(Y_1+Y_2)^{1-t} K (X_1+X_2)^t]\ .
\end{equation*} 
Since $(X,Y)\mapsto \tr[K^*Y^{1-t}K X^t]$ is homogeneous of degree one, this is equivalent to \eqref{lieb1}.
\end{proof}

In the exact same way, one deduces from Theorem~\ref{L2M} and Theorem~\ref{L3M} that the following are true:

\begin{theorem}[Lieb, 1973] \label{L2} For all $0 \leq t \leq 1$, 
\begin{equation*}
(X,Y,K) \mapsto    \tr[ K^*Y^{-1+t}  KX^{-t}]
\end{equation*} 
is jointly convex on $M_n^{++}(\C)\times M_n^{++}(\C)\times M_n(\C)$.  
\end{theorem}

\begin{theorem}[Lieb, 1973]\label{L3}
\begin{equation*}
(X,Y,K) \mapsto    \tr \left[   \int_0^\infty K^*\frac{1}{sI + Y} K \frac{1}{sI + X}{\rm d}s \right]\ 
\end{equation*}
is jointly convex on $M_n^{++}(\C)\times M_n^{++}(\C)\times M_n(\C)$.
\end{theorem}

An argument of the type Lindblad made to deduce the DPI from the joint convexity of the relative entropy can be used to derive Theorems~\ref{L1M}, \ref{L2M} and \ref{L3M} from Theorems~\ref{L1}, \ref{L2} and \ref{L3} respectively; see \cite{C22}. However, we do not need this here: It is simpler to directly prove the monotonicity version.   

In \cite{L73}, Lieb proved that the Theorems Theorems~\ref{L1}, \ref{L2} and \ref{L3}  are all equivalent to one another in the sense that once one has proved any of them, the others follows by simple arguments. The analog of this equivalence for the  monotonicity versions  is much simpler.

Indeed,  Theorem~\ref{L1} and Theorem~\ref{L2} are structurally different in that one is a concavity statement in the two variables $X$ and $Y$, the other is a convexity statement in the three variables $X$, $Y$ and $K$. However, for fixed $X$ and $Y$, the inequalities in the monotonicity version are both comparisons of two quadratic forms in the variable $K$. The Legendre transform provides a duality between quadratic forms that then relates Theorem~\ref{L1M} and Theorem~\ref{L2M}, as will be shown at the end of this lecture. In fact, it will be shown that all three of the fundamental monotonicity theorems are equivalent. 
This equivalence by duality will be significant for the third lecture.

However, before proving Theorems~\ref{L2M} and \ref{L3M}, we explain their relevance  to the questions posed here.  Recall that we are interested in metrics on the space $\mathfrak{S}$ of invertible density matrices. Metrics that have a certain monotonicity property paly a key role in what follows.

In the Introduction, it was noted that  Theorems 2 and 3 of \cite{L73} provided the answers to questions that were not asked until many years later, and that when they were finally asked, 
it was not recognized
that the answers could be found in \cite{L73}.  The questions concerned {\em monotone metrics} on the space of non-degenerate density matrices.  

\subsection{Monotone metrics}

 Let $\rho(t)$ be a differentiable path in the space $\mathfrak{S}$ of invertible $n\times n$ density matrices; i.e., elements of 
 $M_n^{++}(\C)$ with unit trace. Then the derivative $\rho'(t) = K(t)$ is a self-adjoint operator with $\tr[K(t)] =0$.  We may think of 
 $\mathfrak{S}$ as a differentiable manifold, and it is then of interest to equip it with various Riemannian metrics  that have the property that the 
 distance between $\rho_1,\rho_2\in \mathfrak{S}$, $d(\rho_1,\rho_2)$,   decreases under the application of
any quantum operation; i.e., 
\begin{equation}\label{che1}
d(\Phi^\dagger(\rho_1),\Phi^\dagger(\rho_1)) \leq d(\rho_1,\rho_2)\ .
\end{equation}
In such a metric, any quantum operation performed on the states can only make it harder to distinguish between them. See the Introduction to \cite{CFL18} for further discussion of the problem of distinguishing between states, which is basic to quantum communication. 

We may identify the tangent space to $\mathfrak{S}$ at each point $\rho$ with the space of traceless self-adjoint $n\times n$ matrices. If we denote the  quadratic form that specifies a Riemannian metric at $\rho$ by  $\gamma_\rho(K,K)$, then the contractive property \eqref{che1} for such a metric  is equivalent to
\begin{equation*}
\gamma_{\Phi^\dagger(\rho)}(\Phi^\dagger(K),\Phi^\dagger(K)) \leq \gamma_\rho(K,K)\ .
\end{equation*}

It is then natural to ask whether such metrics exist, and if so, what they might be. The same question had been raised  and answered in the classical case  in which the analog of  $\mathfrak{S}$ is the set  $\mathfrak{S}_c$ of $n$ dimensional strictly positive probability vectors 
$p := (p_1,\dots ,p_n)$ with each $p_j >0$ and $\sum_{j=1}^n p_j = 1$.  At each $p\in \mathfrak{S}_c$, 
we may identify the tangent space with the subspace of $\R^n$ consisting of vectors $k = (k_1,\dots,k_n)$ 
with 
$\sum_{j=1}^n k_j =0$.  The analogs of quantum operations are  the  {\em row stochastic  matrices} $P$; i.e., elements of 
$M_n(\R)$ with non-negative and the property that for each $j$, $\sum_{i=1}^n P_{i,j} =1$.   In this setting, 
it is natural to ask for metrics $\gamma$ on  $\mathfrak{S}_c$  with the property that 
\begin{equation*}
\gamma_{Pp}(Pk,Pk) \leq \gamma_p(k,k)
\end{equation*}
for all stochastic $P$ at each $p\in \mathfrak{S}_c$ and for each tangent vector $k$. In 1982, Cencov \cite{Cen82}, building on earlier work of Fisher \cite{F25}, proved that the unique such metric, up to a constant multiple, is the {\em Fisher Information }
\begin{equation*}
\gamma_p(k,k) = \sum_{j=1}^n \frac{k_j^2}{p_j}\ .
\end{equation*}
A decade later, Cencov together with Morozova \cite{MC90}, took up the quantum problem.  In the non-commutative setting, there are many possible 
ways to ``divide by'' a non-degenerate density matrix $\rho$.   Morozova and Cencov came up with several conjectures, all eventually shown to be correct,  
that certain explicit metrics were in fact monotone metrics, although they did not resolve any of these conjectures.  An account of their work can be found in \cite{P96}. 

However, had they known of Theorems 2 and 3 of Lieb's paper, and then recognized them as monotonicity theorems,  they would have had positive solutions to the most important of their conjectures:   These theorems, written in the monotonicity forms \eqref{lieb22} and \eqref{lieb23} show that  for each $0 < t < 1$
\begin{equation}\label{che4}
\gamma_\rho^{(t)}(K,K)  :=  \tr[K \rho^{t-1} K \rho^{-t}]
\end{equation}
and
\begin{equation}\label{che5}
\widehat{\gamma}_\rho(K,K)  :=  \tr\left [\int_0^\infty K \frac{1}{s+\rho} K \frac{1}{s + \rho}{\rm d}s\right] 
\end{equation}
are monotone metrics. As we have already observed in a Legendre transform argument,  the right hand sides are always positive for non-zero $K$, so that these do define  Riemannian metrics, and then the monotonicity is provided by \eqref{lieb22} and \eqref{lieb23}. Note that when $\rho$ and $K$ commute, we have
$$
\gamma_\rho^{(t)}(K,K)    =  \widehat{\gamma}_\rho(K,K)  = \tr\left[\frac{K^2}{\rho}\right]\ ,
$$
as one might expect. {\em Thus these functionals studied by Lieb in 1973 are both natural quantum generalizations of the Fisher Information metric}. 
However, none of the early writers on the subject made the connection with Lieb's theorems, and his paper \cite{L73} is not cited in \cite{P96} which gives the first explicit proof of the fact that \eqref{che4} and \eqref{che5} do in fact define monotone metrics as had been conjectured by Morozova and Cencov. 

The inequality \eqref{che5} can be viewed as an entropy inequality. To see this recall the integral representation for the logarithm of $X\in M_n^{++}(\C)$: 
\begin{equation}\label{logrep}
\log (X) = \int_0^\infty \left(\frac{1}{\lambda+1} - \frac{1}{\lambda +X}\right){\rm d}\lambda\ .
\end{equation}
From here one easily deduces that the right hand side of \eqref{che5} is  the negative of the Hessian of the entropy $S(\rho)$. That is,
$$
\frac{\partial^2}{\partial s \partial t} \tr[ (\rho + s K + t K) \log (\rho  + s K + t K) ]\bigg|_{s=0,t=0}  =  \tr\left [\int_0^\infty K \frac{1}{s+\rho} K \frac{1}{s + \rho}{\rm d}s\right] \ ,
$$
as one can verify using the integral representation for the logarithm \eqref{logrep}. Lesniewski and Ruskai \cite{LR99} showed that all monotone metrics arise in this manner from one of the {\em quasi entropies} that had been introduced by Petz \cite{P85,P86}.

\subsection{The Lieb-Ruskai Monotonicity Theorem}

The Lieb-Ruskai Monotonicity Theorem is a fundamental result proved in \cite{LR74}.  It asserts the monotonicity of an operator function in two variables under completely positive maps.

\begin{theorem}[Lieb and Ruskai 1974]\label{LiRu}  Let $\Phi: M_n(\C) \to M_n(\C)$ be  completely positive. Then for all $K \in M_n(\C)$, and all $X\in M_n^{++}(\C)$
\begin{equation*}
\Phi\left(K^*\frac{1}{X}K\right) \geq \Phi(K)^*\frac{1}{\Phi(X)}\Phi(K)\ .
\end{equation*}
\end{theorem}

Some years later, Choi \cite{Choi80}  proved a generalization. To state this in a  convenient form,  recall that for a matrix $X\in M_n^+(\C)$, the {\em Moore-Penrose} generalized inverse of $X$, denote $X^+$, is the operator obtained from $X$  by replacing all {\em non-zero} eigenvalues of $X$ by their inverses and leaving the kernel as it is. 

\begin{theorem}[Choi 1980]\label{LRC}  A map $\Phi: M_n(\C) \to M_n(\C)$ is $2$-positive if and only if for all $K\in M_n(\C)$ and all $X\in M_n^+(\C)$ such that 
${\rm ker}(X) \subset {\rm ker}(K^*)$,   
\begin{equation}\label{LR7B}
\Phi\left(K^* X^+ K\right) \geq \Phi(K)^*\Phi(X)^+\Phi(K)\ .
\end{equation}
\end{theorem}

Apart from the relaxation of the condition that $X\in M_n^{++}(\C)$, which Choi did not relax, this is an extension of Theorem~\ref{LiRu} since every completely positive map is $2$ positive, but the converse is not true; see, e.g., \cite{Choi72}.  Choi's proof uses a special case of the following well-known lemma on Schur complements.  

\begin{lemma}\label{smlm} For $X,Y\in M_n^{+}(\C)$ and $K\in M_n(\C)$ the following are equivalent:
\begin{equation*}
\left[ \begin{array}{cc} X & K\\K^* & Y\end{array}\right] \geq 0 \ .
\end{equation*}
\begin{equation*}
 {\rm ker}(X) \subseteq {\rm ker}(K^*) \quad{\rm and}\quad  Y \geq K^*X^{+}K\ .
\end{equation*}
\begin{equation*}
 {\rm ker}(Y) \subseteq {\rm ker}(K) \quad{\rm and}\quad  X \geq KY^{+}K^*\ .
\end{equation*}
\end{lemma}

\begin{proof} 
First suppose that $X\in M_n^{++}(\C)$. For any $v\in \C^n$,
$$\left\langle \left(\begin{array}{c} -X^{-1}Kv\\ v\end{array}\right)   \left[ \begin{array}{cc} X & K\\K^*& Y\end{array}\right]  \left(\begin{array}{c} -X^{-1}Kv\\ v\end{array}\right)\right \rangle = \langle v, Yv\rangle - \langle v,K^*X^{-1}Kv\rangle\ .$$
On the other hand,
$$
\left[ \begin{array}{cc} X & K\\K^* & K^*X^{-1}K\end{array}\right] = \left[ \begin{array}{cc} X^{1/2}   & 0\\ K^*X^{-1/2} & 0\end{array}\right] 
\left[ \begin{array}{cc} X^{1/2}   & X^{-1/2}K\\ 0 & 0\end{array}\right] \geq 0\ .
$$
Now suppose only that $X\in M_n^+(\C)$. Then for all $\epsilon>0$, $\left[ \begin{array}{cc} X+\epsilon\one  & K\\K^* & Y\end{array}\right] \geq 0$ and 
$X+\epsilon \one\in M_n^{++}(\C)$. Hence what we have just proved shows that 
\begin{equation}\label{smul2}
\left[ \begin{array}{cc} X & K\\K^* & Y\end{array}\right] \geq 0 \quad \iff \quad Y \geq K^*(X+\epsilon \one) ^{-1}K \qquad{\rm for\ all}\ \epsilon>0\ .
\end{equation}
Evidently, \eqref{smul2} holds for all $\epsilon>0$ if and only if  ${\rm ker}(X) \subseteq {\rm ker}(K^*)$, and in this case
$$
\lim_{\epsilon\to 0} K^*(X+\epsilon \one) ^{-1}K   = K^* X^+K\ .
$$
Therefore, relaxing the condition that $X\in M_n^{++}(\C)$ to $X\in M_n^+(\C)$ yields
${\displaystyle
\left[ \begin{array}{cc} X & K\\K^* & Y\end{array}\right] \geq 0}$ if and only if   ${\rm ker}(X) \subseteq {\rm ker}(K^*)$ and   $Y \geq K^*X^{+}K$.

The upper left and lower right corners of $\left[ \begin{array}{cc} X & K\\K^* & Y\end{array}\right]$ are on an equal footing and therefore it is also true that
${\displaystyle 
\left[ \begin{array}{cc} X & K\\K^* & Y\end{array}\right] \geq 0}$  if and only if  ${\rm ker}(Y) \subseteq {\rm ker}(K)$ and $ X \geq KY^{+}K^*$.
\end{proof}

We next prove Theorem~\ref{LRC} in full generality, without any assumptions of invertibility; see \cite{CMH}.

\begin{proof}[Proof of Theorem~\ref{LRC}]  Suppose that $\Phi$ is $2$-positive, and ${\rm ker}(X) \subset {\rm ker}(K^*)$. By Lemma~\ref{smlm}
$\left[ \begin{array}{cc} X & K\\K^* & K^*X^{+}K\end{array}\right] \geq 0$
and then since $\Phi$ is $2$-positive, $\left[ \begin{array}{cc} \Phi(X) & \Phi(K)\\ \Phi(K)^* & \Phi(K^*X^{+}K)\end{array}\right] \geq 0$. By Lemma~\ref{smlm}
once again, \eqref{LR7B} is valid. 

Next, suppose that  \eqref{LR7B} is valid whenever ${\rm ker}(X) \subset {\rm ker}(K^*)$. We first claim that in this case, 
${\rm ker}(\Phi(X)) \subseteq {\rm ker}(\Phi(K)^*)$. To see this, 
taking $X= \one$, and $K = L^*$ for any $L\in M_n(\C)$, we see that $\Phi$ satisfies
\begin{equation}\label{Schwarz}
\Phi(LL^*) \geq \Phi(L^*)^*\Phi(\one)^+\Phi(L^*) \geq 0
\end{equation}
for all $L\in M_n(\C)$.  In particular, $\Phi$ is positive, and $\Phi(L^*) = \Phi(L)^*$.   For any positive matrix $A$, we have for some $\lambda > 0$, $0 \leq A \leq \lambda \one$, and hence
$0 \leq \Phi(A) \leq \lambda \Phi(\one)$. Hence ${\rm ran}(\Phi(A)) \subseteq {\rm ran}(\Phi(\one))$. Since any $L^*\in M_n(\C)$ can be written as a linear combination of $4$ non-negative matrices, it follows that ${\rm ran}(\Phi(L^*)) \subseteq {\rm ran}(\Phi(\one))$ for all $L\in M_n(\C)$, and consequently it follows from \eqref{Schwarz} that
\begin{equation}\label{KK}
 {\rm ker}(\Phi(LL^*)) \subseteq  {\rm ker}(\Phi(L^*)) \qquad{\rm for\ all}\quad  L \in M_n(\C)\ .
 \end{equation} 
Finally, since  ${\rm ker}(X) \subset {\rm ker}(K^*)$, $X \geq \lambda KK^*$ for some $\lambda > 0$, and then since $\Phi$ is positive and moreover 
satisfies \eqref{Schwarz},  $\Phi(X) \geq \lambda \Phi(KK^*)$, and then by \eqref{KK},
\begin{equation*}
{\rm ker}(\Phi(X)) \subseteq {\rm ker}(\Phi(KK^*)) \subseteq {\rm ker}(\Phi(K)^*)\ .
 \end{equation*} 

 Now suppose that 
 $\left[ \begin{array}{cc} X & K\\K^* & Y\end{array}\right] \geq 0$. By Lemma~\ref{smlm},  ${\rm ker}(X) \subset {\rm ker}(K^*)$ and  $Y \geq K^*X^{+}K$.  By  \eqref{LR7B} and the positivity of $\Phi$,
$$
\Phi(Y) \geq \Phi\left(K^* X^+ K\right) \geq \Phi(K)^*\Phi(X)^+\Phi(K)\ ,
$$
and we have just seen that ${\rm ker}(\Phi(X)) \subseteq{\rm ker}(\Phi(K)^*)$.   
By Lemma~\ref{smlm} again,  
$$\left[ \begin{array}{cc} \Phi(X) & \Phi(K)\\\Phi(K)^* & \Phi(Y)\end{array}\right]  \geq \left[ \begin{array}{cc} \Phi(X) & \Phi(K)\\\Phi(K)^* & \Phi(K^*X^{-1}K)\end{array}\right] \geq 0\ ,$$ and hence $\Phi$ is $2$-positive. 
\end{proof}

Theorem~\ref{LRC} has many consequences. The first is the  following immediate corollary:

\begin{corollary}\label{cl2}  Let $\Phi: M_n(\C) \to M_m(\C)$  be $2$-positive and unital. 
Then for all $K\in M_n(\C)$, 
\begin{equation}\label{LR7K}
\Phi(K^*K) \geq \Phi(K)^*\Phi(K)\ .
\end{equation}
\end{corollary}

Kadison \cite{K52} had proved the inequality \eqref{LR7K} for all {\em self-adjoint} $K$ and all positive $\Phi$, and referred to \eqref{LR7K} as a Schwarz inequality. 
Unital maps that satisfy \eqref{LR7K} are known as {\em Schwarz maps}. 
(The terminology is nearly, but not completely, standard.  Petz \cite[p. 62]{P86} calls any map satisfying \eqref{LR7K} a Schwarz map.)  
 Choi proved  in \cite[Appendix A]{Choi80} that there exist Schwarz maps that are not $2$-positive; e.g., the 
map $\Phi$ on $M_2(\C)$ given by
\begin{equation*}
\Phi(X) = \frac12 X^T  + \frac14 \tr[X] I\ ,
\end{equation*}
where $X^T$ is the transpose of $X$. His construction was further developed in \cite{T85} and \cite[Example 3.6]{HMPB}. 

Corollary~\ref{cl2} is the source of many monotonicity inequalities.
However, in all proofs to come,  we only need  \eqref{LR7B}  (with $\Phi^\dagger$ in place of $\Phi$ for later convenience)  in the {\em tracial} form:
\begin{equation}\label{trform}
\tr\left[\Phi^\dagger \left(K^*X^+K\right)\right] \geq \tr\left[\Phi^\dagger(K)^*\Phi^\dagger(X)^+\Phi^\dagger(K)\right]
\quad{\rm whenever} \quad {\rm ker}(X) \subseteq {\rm ker}(K^*),
\end{equation}
which, when $\Phi^\dagger$ is  quantum operation, and hence trace preserving, is the same as
\begin{equation}\label{trform2}
\tr\left[K^*X^+K\right] \geq \tr\left[\Phi^\dagger(K)^*\Phi^\dagger(X)^+\Phi^\dagger(K)\right]\quad{\rm whenever} \quad {\rm ker}(X) \subseteq {\rm ker}(K^*)\ ,
\end{equation}

One may expect  that \eqref{trform} is valid for a wider class of maps $\Phi^\dagger$ than $2$-positive maps, and this is the case. It was recently proved \cite{CMH} by myself and Alexander M\"uller-Hermes that \eqref{trform2} is satisfied if and only if $\Phi$ satisfies the Schwarz inequality, which while 
\eqref{trform} is satisfied if and only if $\Phi$ is a generalized Schwarz map, as defined in \cite{CMH}.  However, it is the fact that \eqref{trform2} is valid whenever $\Phi$ is $2$-positive and trace preserving that we shall need here. In fact, we will only need to know that this inequality is valid for all quantum operations $\Phi^\dagger$. In short, we we shall see, \eqref{trform2} is in some sense ``the mother of all monotonicity and convexity theorems for trace functions".

\subsection{Operator monotonicity}

A function $f:(0,\infty)\to \R$ is said to be {\em operator monotone increasing} in case for all   $A,B\in M_n^{++}(\C)$, any $n$, 
$A\geq B$ implies $f(A) \geq f(B)$, and $f$ is said to be  {\em operator monotone decreasing} if $-f$ is operator monotone increasing. 

This inverse function $f(x) = x^{-1}$ is operator monotone decreasing. To see this, let $A,B> 0$, and define $X := (A+B)^{-1/2}A(A+B)^{-1/2}$  and $Y :=  (A+B)^{-1/2}B(A+B)^{-1/2}$. 
Evidently, $A\geq B \iff X \geq Y$, and $Y^{-1} \geq X^{-1} \iff B^{-1} \geq A^{-1}$. Since $X+Y =1$, $X$ and $Y$ commute, and then by the spectral theorem $X \geq Y$ implies $Y^{-1}\ \geq X^{-1}$ simply because $f(x) = x^{-1}$ is monotone decreasing in $x$. 

This  example is the source of many other examples: For instance, let $0 < t < 1$.  Then there is the integral representation
$$x^t  =  \frac {\sin(\pi t)} {\pi} \int_0^\infty \lambda^t \left(\frac{1}{\lambda} - \frac{1}{\lambda+x}\right){\rm d}{\lambda}   =
\frac {\sin(\pi t)} {\pi} \int_0^\infty \lambda^t   \frac{x}{\lambda+x}\frac{{\rm d}{\lambda} }{\lambda}
\ ,$$
and now it follows from what we have just proved that $f(x) = x^t$ is operator monotone increasing for all $0 < t < 1$. By a theorem of L\"owner \cite{Lo34}, every operator monotone increasing functions has an integral representation of this general form (see below), and this is the deep part of the theory. Simon's book \cite{S19} contains a beautiful account, with many proofs, some new, of L\"owner's Theorem. However, for the specific examples that arise in this paper, elementary arguments suffice, such as the ones provided just above. 

Nonetheless, we quote L\"owner's Theorem \cite{Lo34} as stated  in Ando and Hiai \cite{AH11}; see also \cite{A78} and \cite{S19}:

\begin{theorem}[L\"owner's theorem]\label{LOT} For $x,\lambda \in (0,\infty)$, define the function $\phi(x,\lambda)$ by
\begin{equation*}
\phi(x,\lambda) := (1 + \lambda) \frac{x}{\lambda + x}\ ,
\end{equation*}
and notice that for each $x$, $\phi(x,\lambda)$ is a bounded function of $\lambda$, so that for any finite positive Borel measure $\mu$  on $(0,\infty)$, all $\beta \in \R$ and all $\gamma \geq 0$,
\begin{equation*}
h(x) :=  \beta + \gamma x + \int_{(0,\infty)} \phi(x,\lambda){\rm d}\mu\ ,
\end{equation*}
is a well defined function on $\R_+$.  The mapping $(\beta,\gamma,\mu) \mapsto h$ is an affine isomorphism  onto the class of operator monotone increasing functions.
\end{theorem}

Note that 
\begin{equation*}
 \frac{x}{\lambda + x}   =  1 -\frac{\lambda}{\lambda+x} \ .
\end{equation*}
By the opening remarks on the inverse function,  it is clear that for each $\lambda$, $\phi(\lambda, x)$ is concave and monotone increasing in $x$, not only as a function of a real variable, but also in the operator sense.   Thus all operator monotone functions $h$  that are real valued on $(0,\infty)$ are also operator concave, and thus $-h$ is operator convex, and monotone decreasing.

\subsection{Proof of the fundamental monotonicity theorems}

As we have indicated in the introduction,  Theorem~\ref{L2M} is equivalent to Theorem~\ref{L1M} by a duality argument, and 
Theorem~\ref{L3M} is equivalent to an integrated form of Theorem\ref{L1M}. This equivalence was pointed out by Hiai and Petz \cite{HP12}  who made 
use of it to give a beautiful proof of these theorems and more.  We give a development of their ideas starting from a duality theorem that not only clarifies the 
origins of the equivalence that they discovered, but also allows the relaxation of some of the  invertibility and non-degeneracy assumptions in their work.

Theorem~\ref{L1M}, Theorem~\ref{L2M} and Theorem~\ref{L3M} can all be viewed as inequalities between a pair of quadratic forms on the Hilbert space 
$\cH_n$ consisting of  $M_n(\C)$ equipped with the Hilbert-Schmidt inner product.  We now explain a basic duality between possibly degenerate quadratic forms on a finite dimensional Hilbert space that we shall apply 

Let $A$ and $B$ be positive definite operators on some finite dimensional Hilbert space $\cH$; that is, $A,B\in \cB^{++}(\cH)$.
Then 
\begin{equation}\label{duality1}
A \leq B \quad \iff\quad B^{-1} \leq A^{-1}\ . 
\end{equation}
One way to see this is to use the readily verified  fact that for all $X\in \cB^{++}(\cH)$,  the Legendre transform of the quadratic form $v\mapsto \frac12\langle v,Xv\rangle$ is $w \mapsto \frac12\langle w,X^{-1}w\rangle$. 
That is,
\begin{equation}\label{duality2}
\sup_{v\in \cH}\{ \Re \langle v,w\rangle - \frac12 \langle v,Xv\rangle\ \} = \frac12\langle w,X^{-1}w\rangle \ .
\end{equation}
If the inequality on the left in \eqref{duality1} is valid, then for all $v,w\in \cH$,
$$ \Re \langle v,w\rangle - \frac12 \langle v,Bv\rangle  \leq \Re \langle v,w\rangle - \frac12 \langle v,Av\rangle \ ,$$
and then taking the supremum over $v$ and using \eqref{duality2}, one obtain the inequality on the right  in \eqref{duality1}.   The same sort of argument, 
or an appeal to the fact that the Legendre transform is involutive, shows that the inequality on the right implies the inequality on the left, thus proving  \eqref{duality1}.

There is a variant of \eqref{duality2} that is valid for $X\in \cB^{+}(\cH)$ and not only for $X\in \cB^{++}(\cH)$. Let $P$ denote the orthogonal projection onto ${\rm ran}(X)$,  the range of of $X$. 
Since $X\in \cB^{+}(\cH)$, $X|_{{\rm ran}(X)}$ is invertible and 
\begin{equation}\label{duality3}
P\left( X|_{{\rm ran}(X)}\right)^{-1}P = X^+\ ,
\end{equation}
where $X^+$ is the generalized inverse of $X$. Then for all $w\in {\rm ran}(X)$, and all $v\in \cH$, 
$$
\Re\langle v,w\rangle - \frac12\langle v,Xv\rangle  = \Re\langle Pv,w\rangle - \frac12\langle Pv,X Pv\rangle\ ,
$$
and therefore, using \eqref{duality2} and \eqref{duality3},
\begin{eqnarray}\label{duality4}
\sup_{v\in \cH} \{ \ \Re\langle v,w\rangle - \frac12\langle v,Xv\rangle \ \} &=& \sup_{x\in {\rm ran}(X)} \{ \Re\langle x,w\rangle - \frac12\langle x,X|_{{\rm ran}(X)} x\rangle \nonumber\\
&=&  \frac12 \langle w,X^+w\rangle\ .
\end{eqnarray}
In the next theorem, we apply these ideas to compare positive operators on two different Hilbert spaces given a linear map form one to the other. 

\begin{theorem}\label{NDual} Let $\cH$ and $\cK$ be two finite dimensional Hilbert spaces, and let $T$ be a linear transformation from $\cH$ into $\cK$. Let $A\in \cB^+(\cH)$ and $B \in \cB^{+}(\cK)$. Suppose that
\begin{equation}\label{duality5}
{\rm ran}(T^*) \subseteq {\rm ran}(A) \quad{\rm and}\quad  {\rm ran}(T) \subseteq {\rm ran}(B)\ .
\end{equation}
Then
\begin{equation}\label{duality6}
\langle v, T^*BTv\rangle \leq \langle v,Av\rangle \qquad{\rm for\ all}\quad v\in \cH
\end{equation}
if and only if 
\begin{equation}\label{duality7}
\langle w, TA^+T^* w\rangle \leq \langle w,B^+w\rangle \qquad{\rm for\ all}\quad w\in \cK\ .
\end{equation}
\end{theorem}

\begin{proof} Suppose that \eqref{duality6} is valid. Then for all $v\in \cH$ and all $w\in \cK$,
\begin{equation}\label{duality8}
\Re\langle v,T^*w\rangle - \frac12 \langle v,Av\rangle\  \leq  \ \Re\langle v,T^*w\rangle - \frac12 \langle v, T^*BTv\rangle\ .
\end{equation} 
Let $P$ denote the orthogonal projection onto ${\rm ran}(A)$. By \eqref{duality5}, $T^*w = PT^*w$ and therefore by \eqref{duality4},
\begin{equation}\label{duality9} 
\sup_{v\in \cH}\{ \Re\langle v,T^*w\rangle - \frac12 \langle v,Av\rangle\ \}  =  \frac12\langle T^*w,A^+T^*w\rangle\ .
\end{equation}

Next,  let $Q$ denote the orthogonal projection onto the range of $B$, and consider the right side of \eqref{duality8}, 
$\Re\langle v,T^*w\rangle = \Re\langle Tv,w\rangle  =  \Re\langle Tv,Qw\rangle$, since by \eqref{duality5}, $Tv\in {\rm ran}(B)$.  Therefore,
\begin{eqnarray}\label{duality10}
\sup_{v\in \cH}\{  \Re\langle v,T^*w\rangle - \frac12 \langle v, T^*BTv\rangle\ \} &=& \sup_{v\in \cH}\{  \Re\langle Tv,Qw\rangle - \frac12 \langle v, T^*BTv\rangle\ \}\nonumber \\
&\leq&  \sup_{x\in \cH}\{  \Re\langle x,Qw\rangle - \frac12 \langle x, B x\rangle\ \}\nonumber \\
&=& \frac12\langle w,B^+w\rangle
\end{eqnarray}
where the last equality comes from \eqref{duality4}. Combining \eqref{duality8}, \eqref{duality9} and \eqref{duality10} shows that \eqref{duality6} implies \eqref{duality7}.

The statement that \eqref{duality7} implies \eqref{duality6}  becomes the statement we have just proves if one swaps $A$ and $A+$, $B$ and $B^+$ and $T$ and $T^*$ and of 
course $\cH$ and $\cK$.  By this symmetry (or the same argument repeated), the proof is complete. 
\end{proof}

 The following construction permits us to write the quadratic forms appearing in the fundamental monotonicity inequalities in manner that permits the application of  Theorem~\ref{NDual}.
Let $\mathcal{H}_m$ denote $M_m(\C)$ equipped with the Hilbert-Schmidt inner product. 
For any $X \in M_m(\C)$, define the operators $L_X$  and $R_X$ on $\mathcal{H}_m$ by $L_X A = XA$ and $R_XA =AX$ respectively.   Note that for any $X,Y\in M_m(\C)$,  $R_X$ and $L_Y$ commute, and that if $X,Y\in M_m^+(\C)$, then $L_Y,R_X\in \cB^+(\cH_m)$. It follows that for any function $f:(0,\infty)\to (0,\infty)$ extended by $f(0)= 0$, one may define the positive semidefinite operator
\begin{equation}\label{Jdef}
{\mathbb J}_f(X,Y) := f(R_X L_Y^{+} )L_Y ,
\end{equation}
for any $Y,X\geq 0$.

Let $\{u_1,\dots,u_m\}$ be an orthonormal basis of $\C^m$ consisting of eigenvectors of $X$; $Xu_j = \lambda_j u_j$, $i=1,\dots,m$. Likewise, let
$\{v_1,\dots,v_m\}$ be an orthonormal basis of $\C^m$ consisting of eigenvectors of $Y$; $Yv_i = \mu_iv_i$. Then define
$$
E_{i,j} = |v_i\rangle\langle u_j| \qquad 1 \leq i,j \leq m\ ,
$$
so that $\{\ E_{i,j}\ :\ 1 \leq i,j \leq m\}$ is an orthonormal basis of  $\mathcal{H}_m$ that simultaneously diagonalizes $R_X$ and $L_Y$. 
Then 
$$
{\mathbb J}_f(X,Y) E_{i,j} = \left(f(\lambda_j\mu_i^+)\mu_i\right) E_{i,j}
$$
where for $x\geq 0$, $x^+$ denotes $x^{-1}$ for $x>0$ and $0$ for $x=0$.  Consequently,
\begin{equation*}
{\mathbb J}_f(X,Y)^+ E_{i,j} = \left(f(\lambda_j\mu_i^+)\mu_i\right)^+ E_{i,j}\ .
\end{equation*}

To make the connection with the fundamental monotonicity inequalities, consider first  the case $f_t(x) = x^t$, $0 < t < 1$. Then for $X,Y\in M_m^{++}(\C)$, and all $K\in M_m(\C)$,
\begin{equation}\label{Jexamp}
 {\mathbb J}_{f_t}(X,Y)(K) = Y^{1-t}K X^{t}\quad{\rm and}\quad    {\mathbb J}_{f_t}(X,Y)^{-1}(K) = Y^{t-1} K X^{-t}\ .
\end{equation}

Consider 
a positive map $\Phi:M_n(\C)\to M_m(\C)$, so that $\Phi^\dagger: M_m(\C) \to M_n(\C)$.   Then by \eqref{Jexamp},  for this choice of $f$  the inequality \eqref{lieb21} of Theorem~\ref{L1M} can be written as
\begin{equation*}
\langle \Phi(K),  {\mathbb J}_{f_t}(X,Y)\Phi(K)\rangle \   \leq \ \langle K, {\mathbb J}_{f_t}(\Phi^\dagger(X),\Phi^\dagger(Y)) K\rangle
\end{equation*}
while the inequality \eqref{lieb22} of Theorem~\ref{L2M} can be written as
\begin{equation*}
 \langle \Phi^\dagger(K), {\mathbb J}_{f_t}^+(\Phi^\dagger(X),\Phi^\dagger(Y))\Phi^\dagger(K)\rangle \  \leq\    \langle K, {\mathbb J}^{-1}_{f_t}(X,Y)  K\rangle
\end{equation*}

We wish to apply Theorem~\ref{NDual} with
$\cH_m$ in place of $\cH$, $\cH_n$ in place of $\cK$, ${\mathbb J}_f(X,Y)$ in place of $A$, ${\mathbb J}_f^+(\Phi^\dagger(X),\Phi^\dagger(Y))$, and $\Phi^\dagger$ in place of $T$.  Once we have verified the condition \eqref{duality5} of  Theorem~\ref{NDual}, we will then know that Theorem~\ref{L1M} and Theorem~\ref{L2M} are equivalent, so that as soon as we have proved one, we have the other. Moreover, for each choice of $f$, not only $f(x) = x^t$, we get an equivalent pair of inequalities. The next choice of immediate interest is 
\begin{equation}\label{logmeanf}
g(x) := \int_0^1 x^t {\rm d} t\ .
\end{equation}
Then for all $X.Y\in M_m^{++}(\C)$ and all $K\in M_m(\C)$,
\begin{equation}\label{Jexamp3}
{\mathbb J}_g(X,Y)(K) = \int_0^1Y^{1-t}K X^{t}{\rm d}t. = \int_0^1 {\mathbb J}_{f_t}(X,Y)(K){\rm d}t  \ .
\end{equation}
Then by since for all $\lambda,\mu>0$, 
$$\int_0^1\lambda^t \mu^{1-t}{\rm d}t =  \frac{\lambda - \mu}{\log \lambda - \log \mu}$$
and
$$
\int_0^\infty \frac{1}{s+\lambda}\frac{1}{s+\mu}{\rm d}s = \frac{\log \lambda - \log \mu}{\lambda - \mu}\ ,
$$
\begin{equation}\label{Jexamp4}
{\mathbb J}_g(X,Y)^{-1}(K) = \int_0^\infty \frac{1}{s+Y}K\frac{1}{s+X}{\rm d}s\ .
\end{equation}

Comparing \eqref{Jexamp4} and inequality \eqref{lieb23} of Theorem~\ref{L3M}, it is clear that the latter can be written as 
\begin{equation*}
 \langle \Phi^\dagger(K), {\mathbb J}_{g}^+(\Phi^\dagger(X),\Phi^\dagger(Y))\Phi^\dagger(K)\rangle \  \leq\    \langle K, {\mathbb J}^{-1}_{g}(X,Y)  K\rangle
\end{equation*}
Then, one we have verified  the condition \eqref{duality5} of  Theorem~\ref{NDual}, we will then know that Theorem~\ref{L3M} is equivalent to 
\begin{equation}\label{L1MJversB}
\langle \Phi(K),  {\mathbb J}_{g}(X,Y)\Phi(K)\rangle \   \leq \ \langle K, {\mathbb J}_{g}(\Phi^\dagger(X),\Phi^\dagger(Y)) K\rangle
\end{equation}
but by the definition \eqref{logmeanf} and then \eqref{Jexamp3}, this is simply in integrated form of Theorem~\ref{L1M}. In fact, in the very last argument in Lieb's paper \cite{L73} in which he 
discusses the equivalence of the original convexity forms of his theorems, it is pointed out that using the integral representation for powers $x^r$, $0 < r < 1$, in terms of the functions $(s+x)^{-1}$,
one can view the inequality \eqref{lieb22} of Theorem~\ref{L2M} as an integrated form of, \eqref{L1MJversB}, the inequality \eqref{lieb23} of Theorem~\ref{L3M}. 

Thus once we know that Theorem~\ref{NDual} applies in our setting, we will know that all three of the fundamental monotonicity theorems  are equivalent. However, more is true: As Hiai and Petz discovered, one can use the equivalence to give a very simple proof of, say, Theorem~\ref{L1M}.

This brings the focus to verifying the condition \eqref{duality5} of Theorem~\ref{NDual} which, with the current choices for $A$, $B$ and $T$  becomes
\begin{equation}\label{duality5B}
{\rm ran}(\Phi) \subseteq {\rm ran}({\mathbb J}_f(X,Y)) \quad{\rm and}\quad  {\rm ran}(\Phi^\dagger) \subseteq {\rm ran}({\mathbb J}_f^+(\Phi^\dagger(X),\Phi^\dagger(Y))\ .
\end{equation}

By what we have noted above, if $X,Y\in M_m^{++}(\C)$, then  ${\mathbb J}_f(X,Y)\in \cB^{++}(\cH_m)$, and hence ${\rm ran}({\mathbb J}_f(X,Y)) = \cH_m$. Therefore, the condition on the left in 
\eqref{duality5B} is satisfied whenever  $X,Y\in M_m^{++}(\C)$.  The condition on the right is more subtle since 
$X,Y\in M_m^{++}(\C)$ does not imply that $\Phi^\dagger(X),\Phi^\dagger(Y)\in M_n^{++}(\C)$.  However, we have the following lemma which is a special case of a lemma proved in \cite{CMH}:

\begin{lemma}\label{CMHlem} Let  $\Phi:M_n(\C)\to M_m(\C)$ be a positive linear map, and let $X,Y\in M_m^{++}(\C)$. Then ${\rm ran}(\Phi^\dagger) \subseteq {\rm ran}(R_{\Phi^\dagger(X)})$ and
${\rm ran}(\Phi^\dagger) \subseteq {\rm ran}(L_{\Phi^\dagger(Y)})$
\end{lemma}

\begin{proof} The statement about $X$ that is to be proved is equivalent to ${\rm ker}(R_{\Phi^\dagger(X)}) \subset {\rm ker}(\Phi)$, To prove this, let $K\in {\rm ker}(R_{\Phi^\dagger(X)})$ We must show that $\Phi(K) =0$, and evidently this is the case if and only if $\langle Y,\Phi(K)\rangle = 0$ for all $Y\in M_m^{+}(\C)$. Since for $Y\in M_m^+(\C)$,
$$
\langle Y,\Phi(K)\rangle = \tr[Y\Phi(K)] = \tr[\Phi^\dagger(Y)K] = \tr[K\Phi^\dagger(Y)]\ ,
$$
we must show that $\tr[K\Phi^\dagger(Y)] = 0$ for all  $Y\in M_m^{+}(\C)$.

Since $X\in M_m^{++}(\C)$, for any $Y\in  M_m^{+}(\C)$, there exists $\lambda>0$ so that $X \geq \lambda Y$. Then since $\Phi^\dagger$ is positive, 
$\Phi^\dagger(X) \geq \lambda \Phi^\dagger(Y)$. 

Since $K\in  {\rm ker}(R_{\Phi^\dagger(X)})$,
$$
0 = K\Phi^\dagger(X) K^* \geq \lambda K\Phi^\dagger(Y) K^* = \lambda\left( K \left(\Phi^\dagger(Y)\right)^{1/2}\right)\left( K \left(\Phi^\dagger(Y)\right)^{1/2}\right)^*\ .
$$
Therefore, $K\left(\Phi^\dagger(Y)\right)^{1/2} = 0$, and consequently $K\Phi^\dagger(Y) =0$. Taking the trace, $\tr[\Phi^\dagger(Y)K] = 0$ as was to be shown. This proves the statemeent about $X$.
The statement about $Y$ may be proved in the same way. 
\end{proof}

\begin{lemma}\label{CMHlem2} Let  $\Phi:M_n(\C)\to M_m(\C)$ be a positive linear map, and let $X,Y\in M_m^{++}(\C)$.  Let ${\mathbb J}_f(X,Y)$ and 
${\mathbb J}_f^+(\Phi^\dagger(X),\Phi^\dagger(Y))$ be defined using \eqref{Jdef}. Then 
\begin{equation}\label{duality21} 
\langle \Phi(K),  {\mathbb J}_f(X,Y)\Phi(K)\rangle \   \leq \ \langle K, {\mathbb J}_f(\Phi^\dagger(X),\Phi^\dagger(Y)) K\rangle \qquad{\rm for\ all}\quad K\in M_n(\C)
\end{equation}
if and only if 
\begin{equation}\label{duality22}
   \langle \Phi^\dagger(K), {\mathbb J}_f^+(\Phi^\dagger(X),\Phi^\dagger(Y))\Phi^\dagger(K)\rangle \  \leq\    \langle K, {\mathbb J}^{-1}_f(X,Y)  K\rangle \qquad{\rm for\ all}\quad K\in M_m(\C)\ .
\end{equation}
\end{lemma}
 
\begin{proof} By Lemma~\ref{CMHlem}, we may apply Theorem~\ref{NDual} with with
$\cH_m$ in place of $\cH$, $\cH_n$ in place of $\cK$, ${\mathbb J}_f(X,Y)$ in place of $A$, ${\mathbb J}_f^+(\Phi^\dagger(X),\Phi^\dagger(Y))$ in place of $B$, and $\Phi^\dagger$ in place of $T$, 
and the result is the equivalence of \eqref{duality21} and \eqref{duality22}. 
\end{proof}

\begin{remark} Lemma~\ref{CMHlem2} was proved by Hiai and Petz \cite{HP12} under the condition that $\Phi^\dagger(\one) > 0$, in which case 
${\mathbb J}_f^+(\Phi^\dagger(X),\Phi^\dagger(Y))$  is invertible whenever $X,Y\in M_n^{++}(X)$. Their proof does not use Theorem~\ref{NDual} but rather a 
simple computation that is valid whenever both ${\mathbb J}_f(X,Y)$ and  ${\mathbb J}_f^+(\Phi^\dagger(X),\Phi^\dagger(Y))$ are invertible. The proof of Lemma~\ref{CMHlem2} as 
stated here was first given in \cite{CMH}, but with a slightly different proof turning on an argument with Schur complements that also sheds useful light on the equivalence; see \cite{CMH}.
\end{remark}

The following is Theorem 5 of Hiai and Petz \cite{HP12} with relaxed conditions on the positive map 
$\Phi$, as extended in  \cite{CMH}.

\begin{theorem}[Hiai, Petz]\label{HPext}  Let $f:(0,\infty) \to (0,\infty)$ be operator monotone,  and define $f(0)= 0$. 
Let  ${\mathbb J}_f$ 
 be defined by \eqref{Jdef}.
Let $\Phi:M_n(\C)\to M_m(\C)$  be a unital Schwarz map. The following inequalities are both valid:

\smallskip
\noindent{\it (a)}  For all $X,Y\in M_{m}^{++}(\C)$, 

$$
\langle \Phi(K),  {\mathbb J}_f(X,Y)\Phi(K)\rangle \   \leq \ \langle K, {\mathbb J}_f(\Phi^\dagger(X),\Phi^\dagger(Y)) K\rangle \qquad{\rm for\ all}\quad K\in M_n(\C)
$$

\smallskip
\noindent{\it (b)}  For all  $X,Y\in M_{m}^{++}(\C)$,  
$$
  \langle \Phi^\dagger(K), {\mathbb J}_f^+(\Phi^\dagger(X),\Phi^\dagger(Y))\Phi^\dagger(K)\rangle \  \leq\    \langle K, {\mathbb J}^+_f(X,Y)  K\rangle \qquad{\rm for\ all}\quad K\in M_m(\C)\ .
  $$
\end{theorem}

 \begin{proof}[Proof of Theorem~\ref{HPext}]  By Lemma~\ref{CMHlem2}, the inequalities in parts {\it (a)} and {\it (b)} are equivalent, and hence it suffices to prove the inequality in {\it (a)}. 
 Note the the inequality in {\it (a)}, unlike that in {\it (b)}, is additive in $f$. Therefore, using 
 the L\"owner theorem, Theorem~\ref{LOT}, which gives an integral representation of all operator monotone functions, Hiai and Petz show that it suffices to do this for  the special case
 \begin{equation*}
f(x) :=  \beta + \gamma x +    \frac{x}{t + x}  
 \end{equation*}
 with $\beta,\gamma,t \geq 0$.

  To prove {\it (a)} for this choice of $f$ it suffices to prove
 \begin{equation}\label{ph7}
\Phi^\dagger L_Y \Phi \leq L_{\Phi^\dagger(Y)}\ ,\quad   \Phi^\dagger R_X \Phi \leq R_{\Phi^\dagger(X)}
\end{equation}
and
\begin{equation}\label{ph8}
\Phi^\dagger   \frac{R_X}{t + R_XL_{Y^{+}}} \Phi  \leq  \frac{R_{\Phi^\dagger(X)}}{t + R_{\Phi^\dagger(X)}L_{\Phi^\dagger(Y)^{+}}}\ .
\end{equation}
For any $K\in M_{{n}}(\C)$, using the Schwarz inequality 
we have 
\begin{align*}\langle K, \Phi^\dagger L_Y \Phi K\rangle &= \tr[\Phi(K)^* Y \Phi(K)] \\
 &\leq \tr[\Phi(KK^*)Y] = \tr[KK^*\Phi^\dagger(Y)] = \langle K, L_{\Phi^\dagger(Y)} K\rangle\ ,
 \end{align*}
 and this proves the first inequality in \eqref{ph7}. The proof of the second is entirely analogous. To prove \eqref{ph8}, note that by the equivalence of the inequalities in {\it (a)} and {\it (b)}, it suffices to show that,
 \begin{equation*}
\tr\left[  \Phi^\dagger(K)^*  \left(\frac{R_{\Phi^\dagger(X)}}{t + R_{\Phi^\dagger(X)}L_{\Phi^\dagger(Y)}^{+}}  \right)^{+}\Phi^\dagger(K)\right]   \leq \tr\left[K^*\left(   \frac{R_X}{t + R_XL_Y^{-1}} \right)^{-1}K\right]\ .
 \end{equation*}
 A simple computation shows that 
 $$
  \left(\frac{R_{\Phi^\dagger(X)}}{t + R_{\Phi^\dagger(X)}L_{\Phi^\dagger(Y)}^{+}}  \right)^{+} = (t R_{\Phi^\dagger(X)}^+ +  L_{\Phi^\dagger(Y)}^+)\ ,
  $$
  and of course
  $$
  \left(   \frac{R_X}{t + R_XL_Y^{-1}} \right)^{-1} =  tR_X^{-1} + R_Y^{-1}\ .
  $$
  Hence all that remain to be shown is that 
 \begin{align}\label{equ:finIneq}
 &t \tr[ \Phi^\dagger(K) \Phi^\dagger(X)^{+} \Phi^\dagger(K^*) ] +  \tr[ \Phi^\dagger(K^*) \Phi^\dagger(Y)^{+} \Phi^\dagger(K)] \nonumber\\
  &\quad\quad\quad\leq t \tr[KX^{-1}K^*] + \tr[K^*Y^{-1}K]\ ,
 \end{align}
 for all $K\in M_{{m}}(\C)$. Since $\Phi$ is a Schwarz map, by \cite[Theorem 4]{CMH},  \eqref{trform2} is satisfied 
 and hence we have both
 \begin{equation*}
 \tr[ \Phi^\dagger(K) \Phi^\dagger(X)^{+} \Phi^\dagger(K^*) ]   \leq  \tr[KX^{-1}K^*]  
 \end{equation*}
 and
 \begin{equation*}
\tr[ \Phi^\dagger(K^*) \Phi^\dagger(Y)^{+} \Phi^\dagger(K)] \leq  \tr[K^*Y^{-1}K]\ ,
 \end{equation*}
 and \eqref{equ:finIneq} follows. 
 \end{proof}
 
 \begin{remark} Hiai and Petz assumed the $\Phi$ was unital, completely positive and such that $\Phi^\dagger(\one)> 0$. Their proof extends without change to the case in which  
 compete positivity is relaxed to $2$ positivity, since then one still has the crucial trace inequality. The proof here, as in Uhlmann's original proof \cite{Uh77} of Theorem~\ref{L2M}, 
 uses only that $\Phi$ is a Schwarz map. It was shown in \cite{CZ23} that this is optimal; Theorem~\ref{HPext} does not hold for any larger class of positive maps $\Phi$. 
 
 On the other hand, for our purposes, and most purposes in mathematical physics, it suffices to know that the inequalities of Theorem~\ref{HPext} are valid when 
 $\Phi^\dagger$ is a quantum operation, or, what is the same thing, when $\Phi$ is unital and completely positive. The proof we have given of Theorem~\ref{HPext} is not 
 completely self contained, as it uses 
 \cite[Theorem 4]{CMH}. However, when $\Phi^\dagger$ is completely positive (even $2$-positive) and trace preserving, Theorem~\ref{LRC}, proved here in full, provides the inequality
 \eqref{trform2}  that we need in the final step of the proof.  
 \end{remark}

Before proving Theorem~\ref{HPext}, we apply it to prove Theorems~\ref{L1M}, \ref{L2M} and \ref{L3M}.    

\begin{proof}[Proof of Theorems~\ref{L1M} and \ref{L2M}]    Consider $f(x) = x^r$, $0 < r < 1$,  which as we have noted is operator monotone with an explicit integral representation of L\"owner type.   Then
 in the case of $f(x) = x^r$, $0 < r < 1$,  
 $$
 {\mathbb J}_f(\Phi^\dagger(X), \Phi^\dagger(Y))K =  \Phi^\dagger(Y)^{1-r} K \Phi^\dagger(X)^{r}  
 $$
 and
 $$
 {\mathbb J}_f(X,Y)(\Phi(K)) = Y^{1-r} \Phi(K) X^{r}\ .
 $$
 Then since by {\it (a)},
 $$
 \langle \Phi(K), Y^{1-r} \Phi(K) X^{r}\rangle \leq \langle K, \Phi^\dagger(Y)^{1-r} K \Phi^\dagger(X)^{r} \rangle\ ,
 $$
 we have
 \begin{equation*}
 \tr[\Phi(K)^* Y^{1-r} \Phi(K) X^{r}] \leq \tr[ K^* \Phi^\dagger(Y)^{1-r} K \Phi^\dagger(X)^{r}]\ ,
 \end{equation*}
 which proves Theorem~\ref{L1M}. 
  Likewise, from {\it (b)} we get
\begin{equation*}
 \tr[\Phi^\dagger(K)^* (\Phi^\dagger(Y)^+)^{1-r} \Phi^\dagger(K) (\Phi^\dagger(X)^+)^{r}] \leq \tr[ K^* Y^{r-1} K X^{-r}]\ ,
 \end{equation*}
 and this prove Theorem~\ref{L2M}. 
 
 To prove Theorem~\ref{L3M}, take $f(x) := \int_0^1 x^r {\rm d}r$, so that
 $$
 {\mathbb J}_f(X,Y)(K) = \int_0^1 Y^{1-r} K X^r {\rm d}r\ .
 $$
 Notice that if $v$ is an eigenvector of $Y$, and $W$ is an eigenvector of $X$, $|v\rangle \langle w|$ is an eigenvector of  ${\mathbb J}_f(X,Y)$. In this way we get an explicit diagonalization of $ {\mathbb J}_f(X,Y)$ from which it follows that
 $$
  {\mathbb J}_f(X,Y)^{-1}(K)   = \int_0^\infty \frac{1}{\lambda +Y} K \frac{1}{\lambda +X} {\rm d}\lambda\ .
 $$
 Then {\it (b)} gives us Theorem~\ref{L3M}. 
 \end{proof}

 At this point we have given self-contained proof of all of the entropy  inequalities that we shall use in the following lectures.  As will be seen, it is the ideas and constructions going into the proofs, as well as the inequalities themselves, that will be useful in what follows. 
 In particular,  operator ${\mathbb J}_f(X,Y) = f(R_X L_Y^{+} )L_Y $
defined in \eqref{Jdef} will, for appropriate choices of $f$, come up again in the construction of the quantum mass transportation metric that we discuss in the next two lectures.

%% file: chapterC2.tex
%
%
%
\section{ Lecture Two:  Reversible QMS as Gradient flow for Relative Entropy }


\abstract{\footnotesize According to the Data Processing Inequality, for every quantum dynamical semigroup $\{\cP_t^\dagger\}_{t\geq 0}$, any invariant state $\sigma$ for this semigroup, and any other state $\rho$, $D(\cP_t^\dagger \rho||\sigma)$ 
 is monotone decreasing in $t$. When can the evolution described by semigroup be viewed as gradient flow  for the relative entropy with respect to some metric?  This is not always the case. We discuss the necessary conditions and the known sufficient conditions, giving details of the construction of the metrics. Many people have contributed to this subject, and while other approaches will be discussed, the focus will be on the approach developed by myself and  Jan Maas.}

\subsection{Gradient flow on $\mathfrak{S}$}   

Recall that $\Dens$, the space of strictly positive density matrices,  is a relatively open subset of $\{ A \in M_n^+(\C)\  :  \tr[A] = 1 \}$, and we may identify the tangent space $T_{\rho} \Dens$
at each point $\rho \in \Dens$ with 
\begin{equation}\label{TS}
 \mathcal{V} := \{ A \in M_n(\C)  : A =A^* \quad{\rm and}\quad  \tr[A] = 0 \}\ .
 \end{equation}
 The cotangent space $T_{\rho}^\dagger \Dens$ may also be identified with $\mathcal{V}$ through the duality pairing $\ip{A,B} = \tr[A B]$ for $A, B \in  \mathcal{V}$.

Let $\{g_{\rho}\}_{\rho \in \Dens}$ be a Riemannian metric on $\Dens$, i.e., a collection of positive definite bilinear forms $g_{\rho} : T_{\rho}\Dens \times T_{\rho}\Dens \to \R$ depending smoothly on $\rho \in \Dens$. 
Consider the associated operator $\cG_{\rho} :  T_{\rho}\Dens \to  T_{\rho}^{\dagger}\Dens$ defined by 
$$\ip{A, \cG_{\rho} B} = g_{\rho}(A, B)$$ 
for $A, B \in T_{\rho} \Dens$. 
Clearly, $\cG_{\rho}$ is invertible and self-adjoint with respect to the Hilbert--Schmidt inner product.
Define $\cK_{\rho} :  T_{\rho}^{\dagger}\Dens \to T_{\rho} \Dens$ by $\cK_{\rho} = (\cG_{\rho})^{-1}$, so that 
\begin{align}\label{eq:metric-K}
g_{\rho}(A,B) = \ip{ A, \cK_{\rho}^{-1} B} \ .
\end{align}

The contravariant metric tensor consists of the collection of positive definite bilinear forms $g^\dagger_{\rho} : T_{\rho}^\dagger\Dens \times T_{\rho}^\dagger\Dens \to \R$ depending smoothly on $\rho \in \Dens$ given by
\begin{align*}
g_{\rho}^\dagger(A,B) = \ip{ A, \cK_{\rho} B} \ .
\end{align*}

Equipping  $T_{\rho}\Dens$ and $T_{\rho}^\dagger\Dens$ with these inner products, making them real Hilbert spaces, 
$\cG_{\rho} :  T_{\rho}\Dens \to  T_{\rho}^{\dagger}\Dens$ is  an isometry. We could have started with the the inner product on $T_{\rho}^\dagger\Dens$ specified by  the positive operator $\cK_{\rho}$, and then defined $\cG_{\rho}$ to be its inverse. In fact, this is often the most convenient way to specify the metrics with which  we shall work.

For a smooth functional $\cF : \Dens \to \R$ and $\rho \in \Dens$, its \emph{differential} $\rmD\cF(\rho) \in T_{\rho}^{\dagger}\Dens$ is defined by 
$$\ip{A,\rmD \cF(\rho)} =\lim_{\eps \to 0} \eps^{-1}(\cF(\rho + \eps A) - \cF(\rho) ) $$
 for $A \in T_{\rho}\Dens$. Note that this is independent of the Riemannian metric $g_{\rho}$.
However, the \emph{gradient} $\nabla_g \cF(\rho) \in T_{\rho}\Dens$ depends on the Riemannian metric through the duality formula  that defines it, namely
$$g_{\rho}(A, \nabla_g \cF(\rho)) = \ip{A, \rmD \cF(\rho)}$$ for $A \in T_{\rho}\Dens$. It follows that $\cG_{\rho}\nabla_g \cF(\rho) = \rmD \cF(\rho)$, or equivalently
\begin{align*}
	\nabla_g \cF(\rho) = \cK_{\rho} \rmD \cF(\rho) \ .
\end{align*}
Given a smooth function $t\mapsto \rho(t)$ with values in $\Dens$, at each $t$,
$$
\lim_{h\to 0} \frac1h(\rho(t+h) - \rho(t)) =: \partial_{t}\rho(t) \in T_{\rho(t)}\Dens\ .
$$
The gradient flow equation $\partial_{t}\rho(t) = - \nabla_g \cF(\rho(t))$ takes the form 
\begin{align*}
\partial_{t}\rho(t)(t) = - \cK_{\rho} \rmD \cF(\rho) \ .
\end{align*}

Now consider a QMS $(\cP_t)_{t\geq 0}$ with generator $\cL$ so that $\cP_t = e^{t\cL}$. Suppose that $\sigma$ is the unique invariant state for this semigroup. Given $\rho_0\in \Dens$, define the curve
$$
\rho(t) := \cP_t^\dagger \rho_0\ .
$$
Then by the DPI,  $t\mapsto D(\rho(t)||\sigma)$ is monotone decreasing.  This is an obvious necessary condition for the evolution specified by 
$(\cP_t^\dagger)_{t\geq 0}$ to be gradient flow for the function $\cF(\rho) := D(\rho ||\sigma)$ on $\Dens$.   There is however a further necessary condition: each
$\cP_t$ must be self adjoint with respect to an inner product associated with the entropy. 

Recall that 
$$
\frac{\partial^2}{\partial s \partial t} \tr[ (\rho + s A + t B) \log (\rho  + sA + t B) ]\bigg|_{s=0,t=0}  =  \tr\left [\int_0^\infty A \frac{1}{\lambda+\rho} B \frac{1}{\lambda + \rho}{\rm d}\lambda\right] \ ,
$$
Define $\cD_\rho$ by
\begin{equation}\label{logdif}
\cD_\rho(A) :=  \int_0^\infty \frac{1}{\lambda +\rho} A \frac{1}{\lambda + \rho}{\rm d}\lambda = \frac{\rm d}{{\rm d}s} \log (\rho + sA)\bigg|_{s=0}
\end{equation}
so that 
$$
\langle A, \cD_\rho B\rangle =  \tr\left [\int_0^\infty A \frac{1}{\lambda+\rho} B \frac{1}{\lambda + \rho}{\rm d}\lambda\right] \ .
$$
In this case it is easy to compute $\cD_\rho^{-1}$:
$$
\cD_\rho^{-1}(A) = \int_0^1 \rho^{1-t} A \rho^{t}{\rm d}t =   \frac{\rm d}{{\rm d}s} \exp (\log\rho + sA)\bigg|_{s=0}\ .
$$

We then have the following non-degenerate quadratic forms:
$$\langle A, \cD_\rho(B)\rangle \quad{\rm and}\quad  \langle A, \cD^{-1}_\rho(B)\rangle\ .$$
In this geometric context, these quatratic forms are defined in the space ${\mathcal V}$ defined in \eqref{TS},  but there is the obvious way to extend them to all of $M_n(\C)$, and this is useful to do. In fact, the extension of the inner prodcut  $\langle A, \cD^{-1}_\rho(B)\rangle$ all of $M_n(\C)$ has a name: It is the
{\em Boguliobov--Kubo--Mori} (BKM) inner product \cite{KTH91}  specified by $\rho$ which is given by
\begin{equation}\label{BKMipdef}
\langle A,B\rangle_{{\rm BKM}} := \langle A, \cD^{-1}_\rho(B)\rangle =\int_0^1 \tr\left[ A^* \rho^{1-t} B\rho^{t}\right]{\rm d}t\ .
\end{equation}
The following theorem was proved by Maas and myself \cite{CM20}, and it is a quantum version of a result proved in the Markov chain setting by Dietert \cite{Die15}. 

\begin{theorem}\label{necBKM}  Let $(\cP_t)_{t\geq 0}$ be an ergodic QMS with generator $\cL$ and invariant state $\sigma\in \Dens_+$. 
If there exists a continuously differentiable Riemannian metric $(g_\rho)$ on $\Dens_+$ such that the quantum master equation $\partial \rho = \cL^\dagger \rho$ is the gradient flow equation for $\Ent_\sigma$ with respect to $(g_\rho)$, then each $\cP_t$ is self-adjoint with respect to the BKM inner product associated to $\sigma$. 
\end{theorem}

\begin{proof}[Proof of Theorem~\ref{necBKM}]

Suppose that there is such a metric, and let  operator $\cK_\rho$ be defined in terms of it  by \eqref{eq:metric-K}.
Since $\rmD D(\rho ||\sigma) = \log \rho - \log \sigma$, the gradient flow equation $\partial_t \rho = - \cK_{\rho} \rmD \Ent_{\sigma}(\rho)$ is 
\begin{align*}
	\cL^\dagger \rho= - \cK_\rho (\log \rho - \log \sigma) \ . 
\end{align*}
Applying this identity to $\rho_\eps = \sigma + \eps A$ for $A \in \cA_0$, and differentiating at $\eps = 0$, we obtain using \eqref{logdif}
\begin{align*}
	\cL^\dagger A = - \cK_\sigma \left(  \int_0^\infty \frac{1}{\lambda +\sigma} A \frac{1}{\lambda + \sigma}{\rm d}\lambda\right) = -\cK_\sigma\cD_\sigma(A)\ ,
\end{align*}
or equivalently,
$$
\cL  A = -  \left(  \int_0^\infty \frac{1}{\lambda +\sigma} \cK_\sigma(A) \frac{1}{\lambda + \sigma}{\rm d}\lambda\right)  = -\cD_\sigma\cK_\sigma(A)\ ,
$$
Consequently, for $A, B \in {\mathcal V}$, with ${\mathcal V}$ given by \eqref{TS},
\begin{align*}
	\ip{\cL A, B}_{L^2_{\rm BKM}(\sigma)}
	= - \tr[ A^* \cK_\sigma B] \ .
\end{align*}

As  the operator $\cK_\sigma$ is self-adjoint with respect to the Hilbert-Schmidt scalar product, this implies the result.
\end{proof}

A converse has very recently been proved by Brooks and Maas \cite{BM23}:  The condition of BKM self adjointness is sufficient as well as necessary. 
At this stage, the construction of the Riemannian metric is somewhat implicit, and it is not presently clear in what sense it might be a transport 
metric, or how it might be used to address the questions raised in the first lecture. 

However, for a physically important sub-class of QMS semigroups for which each $\cP_t$ is self-adjoint with respect to the 
BKM inner product for the invariant state, one can construct the metric very explicitly, and use this construction to prove 
inequalities governing the rate of approach to the invariant state, among other things. The basis of this construction is a structure theorem for 
the generators of this class of processes.   In the next section, we discuss structure theorems for various classes of 
QMS generators, focusing on results we shall use to construct our metrics. 

\subsection{Inner products on $M_n(\C)$ associated to $\sigma\in \Dens$ and the structure of QMS Generators}

Work of  Lindblad \cite{Lin76} and Gorini, Kossakowski and Sudarshan \cite{GKS76} has provided a structure theorem for the set of QMS generators. It turns out that if one imposes certain
self-adjointness conditions on the QMS, the structure theorem becomes much more specific. We have already seen in Theorem~\ref{necBKM} 
 that if $\cP_t$ is a QMS such that $\cP_t^\dagger$ is given by 
gradient flow for the relative entropy functional $D(\rho||\sigma)$, where $\sigma$ is the invariant state for $\cP_t^\dagger$, then $\cP_t$ must be self adjoint for the 
BKM inner product defined in \eqref{BKMipdef}. Unfortunately, there does not presently exist a useful  structure theorem for the set of BKM self-adjoint QMS generators, 
but for the smaller class of GNS self-adoint QMS generators, introduced below, there is such a structure theorem due to Alicki \cite{Al78}.

 We now discuss various inner products on $M_n(\C)$ determined by a state $\rho\in \Dens$, and we closely follow \cite{AC20}. Let $\mathcal{P}[0,1]$ denote the set of probability measures on the interval $[0,1]$.  Notice that for each $s\in [0,1]$,
$$
\tr[B^* \sigma^{1-s}A \sigma^s] =  \tr[ ( \sigma^{(1-s)/2}B \sigma^{s/2})^*  \sigma^{(1-s)/2}A \sigma^{s/2}]\ ,
$$
and the right hand side is strictly positive when $B=A \neq0$.

\begin{definition} For each  $m\in\mathcal{P}[0,1]$, $\langle \cdot,\cdot\rangle_m$ denotes the inner product on $M_n(\C)$ given by
\begin{equation}
    \langle B,A \rangle_m = \tr[B^\ast\mathcal{M}_m(A)] \quad \text{where} \quad \mathcal{M}_m(A) = \int_0^1 \sigma^{s} A \sigma^{1-s} \; \mathrm{d} m (s)  \ .
    \label{integral_M}
\end{equation}
\end{definition}

The {\em Gelfand-Naimark-Segal} (GNS) inner product corresponds to $m = \delta_0$, the point mass at $s=0$. Other cases are known by name. Taking $m = \delta_{1/2}$ yields the {\em Kubo-Martin-Schwinger} (KMS) inner product
\begin{equation*}
\langle B,A\rangle_{{\rm KMS}} = \tr[B^*\sigma^{1/2}A\sigma^{1/2}]\ .
\end{equation*}
Taking $m$ to be uniform on $[0,1]$ yields the  Bogoliubov-Kubo-Mori  inner product 
$\langle B,A\rangle_{{\rm BKM}}$ already defined in \eqref{BKMipdef}.

For $\sigma = n^{-1}\one$, the normalized identity, all choices of $m$ reduce to the normalized Hilbert-Schmidt inner product.  

Throughout what follows, $\mathfrak{H}$ always denotes the Hilbert space $M_n(\C)$ equipped with this inner product,
\begin{equation*}
\langle B,A\rangle_{\mathfrak{H}} = \frac1n\tr[B^*A]\ .
\end{equation*}
Let  $\mathcal{L}(M_n(\C))$ denote the linear operators on $M_n(\C)$, or, what is the same thing in this finite dimensional setting, on $\mathfrak{H}$.  A dagger is always used to denote the adjoint with respect to the inner product on $\mathfrak{H}$, That is,  for $\Phi\in \mathcal{L}(M_n(\C))$,  $\Phi^\dagger$ is defined by
\begin{equation*}
\langle \Phi^\dagger(B),A\rangle_{\mathfrak{H}} =  \langle B, \Phi(A)\rangle_{\mathfrak{H}}
\end{equation*}
for all $A,B$, which is consistent with our previous notation since the normalization has no effect on the adjoint.  

 We can also make $\mathcal{L}(M_n(\C))$ into a Hilbert space by equipping it with the normalized Hilbert-Schmidt inner product. 
Throughout what follows, this  Hilbert space is denoted by $\widehat{\mathfrak{H}}$. The following formula for the inner product in 
$\widehat{\mathfrak{H}}$ is often useful. Let  $\{F_{i,j}\}_{1\leq i,j\leq N}$ be any  orthonormal basis for $\mathfrak{H}$. Then for $\Phi$ and $\Psi\in \widehat{\mathfrak{H}}$, 
\begin{equation*}
\langle \Psi,\Phi\rangle_{\widehat{\mathfrak{H}}} = \frac{1}{n^2}\sum_{i,j=1}^n  \langle \Psi(F_{i,j}), \Phi(F_{i,j})\rangle_{\mathfrak{H}} \ .
\end{equation*}

To each $\sigma\in \Dens_+$, there corresponds the {\em modular operator} $\Delta_\sigma$ on $\mathfrak{H}$ defined by 
\begin{equation*}
\Delta_\sigma A = \sigma A \sigma^{-1}\ .
\end{equation*}
Consider any  orthonormal basis $\{u_1,\dots,u_n\}$ of $\C^N$ consisting of eigenvectors of $\sigma$, so that for $1 \leq j \leq n$, $\sigma u_j = \lambda_j u_j$.  For $1\leq i,j \leq n$, define 
$
E_{i,j} = \sqrt{n} | u_i \rangle \langle u_j|
$, so that $\{E_{i,j}\}_{1\leq i,j\leq n}$ is an orthonormal basis of $\fH$.  A simple computation shows that for each $i,j$,
\begin{equation*}
\Delta_\sigma E_{i,j} = \lambda_i\lambda_j^{-1} E_{i,j}\ .
\end{equation*}
Thus $\Delta_\sigma$ is diagonalized, with positive diagonal entries, by an orthonormal basis in $\mathfrak{H}$. It follows that $\Delta_\sigma$ is a positive operator on $\mathfrak{H}$.   
Using the Spectral Theorem, we may then define $\Delta_\sigma^{s}$ for all $s$. This provides another way to write the operator ${\mathcal M}_m$ defined in  \eqref{integral_M}:
\begin{equation*}
{\mathcal M}_m = \left( \int_0^1 \Delta_\sigma^s {\rm d}m\right) R_\sigma
\end{equation*}
where $R_\sigma$ denotes right multiplication by $\sigma$, also a positive operator on $\mathfrak{H}$ that commutes with $\Delta_\sigma$, 
and hence with $\int_0^1 \Delta_\sigma^s {\rm d}m$. 

This may well look familiar. If we define $f_m(x) = \int x^s {\rm d}m$, then 
$$
{\mathcal M}_m  = {\mathbb J}_f(\sigma,\sigma)
$$
where ${\mathbb J}_f(X,Y)$ is the operator introduced in connection with the Hiai-Petz method for proving trace inequalities.

As is well-known, both the set of completely positive maps and the set of QMS generators  are convex cones.  For completely positive maps this is obvious. For the case of QMS generators, consider a QMS    $\{\cP_t\}_{t\geq 0}$  with generator 
$$\cL := \lim_{t\to0}\frac1t(\cP_t - I) $$
 Note that $\cL \one = 0$, and $\cP_t = e^{t\cL}$. If $\cL_1$ and $\cL_2$ are two  generators of completely positive semigroups, then
\begin{equation}\label{trotter}
e^{t(\cL_1 + \cL_2)} = \lim_{k\to\infty} (e^{(t/k)\cL_1}e^{(t/k)\cL_2})^k
\end{equation}
is completely positive,  and hence the set of generators of completely positive semigroups is closed under addition, and also evidently under multiplication 
by non-negative real numbers.  Since if $\cL_j\one = 0$ for $j=1,2$, then $(\cL_1+\cL_2)\one =0$, it follows easily that the set of all 
QMS semigroup generators is a convex cone.  However, it is not a pointed cone: Let $H\in M_n(\C)$ be self-adjoint, and define $\cL(A) = i[H,A]$. 
Then $e^{t\cL}(A) = e^{itH}A e^{-it H}$ is a QMS, and evidently both $\cL$ and $-\cL$  are QMS generators.

There are some obvious examples of semigroups of completely positive operators on $M_n(\C)$. For instance if $\Phi$ is completely positive, so it $e^{t\Phi}$. 
Likewise, for any $G \in M_n(\C)$, so is each $A \mapsto e^{-tG^*} A e^{-tG}$, which has the generator
$$
\cL(A) = -G^*A - AG\ .
$$
By what we have noted above, for any completely positive map $\Phi$ on $M_n(\C)$ and any $G\in M_n(\C)$,
\begin{equation}\label{lind0}
\cL(A) := \Phi(A) -G^*A - AG
\end{equation}
generates a semigroup of completely positive operators. It is a QMS generator if and only if $\cL(\one) = 0$, and this in turn is the case if and only if
$\Phi(\one) = G^* +G$.

Write $G := K+iH$ where $H$ and  $K$ are self adjoint.  Then 
\begin{equation}\label{lind2}
\cL(A) =  \Phi(A) - \frac12(\Phi(\one)A + A \Phi(\one)) + i[H,A]\ .
\end{equation}
It turns out that this class of examples is all that there is: By a 1976 theorem
 of Lindblad \cite{Lin76} and Gorini, Kossakowski and Sudarshan \cite{GKS76}, every QMS generator  $\cL$  on $M_n(\C)$ has this exact  form.   
 Moreover, since every completely positive map $\Phi$ has a Kraus representation
 \begin{equation}\label{kraus}
 \Phi(A) = \sum_{j=1}^m V_j^*AV_j\  
\end{equation}
for some $\{V_1,\dots,V_m\}\subset M_n(\C)$,
 \eqref{lind2} can be rewritten as
\begin{equation}\label{lind2b}
 \cL A  =  i [H,A] + \sum_{j =1}^m  V_j^* [A, V_j] +  [V_j^* , A] V_j \ , 
 \end{equation}
 and consequently,
 \begin{equation*}
  \cL^\dagger \rho  = - i [H,\rho] + \sum_{j=1}^m  [V_j, \rho V_j^*] +  [V_j \rho, V_j^*]
\end{equation*}
(All of this will be proved in the next section of this lecture.)

The terms $\Phi$ and $G$ in the decomposition of $\cL$ in \eqref{lind0} are not uniquely determined by $\cL$.  Indeed, consider any CP map $\Phi$ with the Kraus representation \eqref{kraus}. Then for any choice of complex numbers $c_1,\dots,c_m$, define $W_j := V_j + c_j\one$, and 
$$
\Psi(A) := \sum_{j=1}^m W_j^*AW_j + \sum_{j=1}^m|c_j|^2 A\ .
$$
Then a simple computation shows that 
$$
\Phi(A) = \Psi(A) - G^*A - AG \qquad{\rm where}\qquad G = \sum_{j=1}c_jW_j\ .
$$
Consequently, in the Lindblad form \eqref{lind2b}, the ``Hamiltonian part'' $ i [H,A] $ is not uniquely determined by $\cL$.

The situation is better if we impose some sort of self adjointness requirement on $\cL$. However, for general $m$, and for $\cL$ in the form \eqref{lind0},
it is possible for $\cL$ to be self adjoint but not $\Phi$, and vice-versa -- even for any choice of the non-unique decomposition. 
Hence in general,
 the problem of determining the structure of the cone of completely positiver maps that are self adjoint with respect to $\langle\cdot,\cdot\rangle$   
 is different from the problem of determining the structure of the cone QMS generators the are self adjoint with respect to some $\langle\cdot,\cdot\rangle_m$.    
 There is, however,  an important family of cases in which there is a canonical decomposition into self adjoint components.  
 This is when $m  = \delta_s$, the point mass at $s\in [0,1]$, which includes the GNS and KMS inner products, but not the BKM inner product.

\subsection{Structure of QMS generators}

In this section we begin to explain the structure theory for various classes of QMS generators. We begin with some useful tools that are used in the proof of the 
structure theorems, but are more widely useful. Again in this section we closely follow \cite{AC20}.   

Recall that $\fH$ denotes $M_n(\C)$ equipped with the normalized Hilbert-Schmidt inner product.  Let $\mathcal{L}(\fH)$ denote the space of linear transformation on $\fH$.  We make 
$\mathcal{L}(\fH)$ into a Hilbert space $\widehat{\fH}$ by equipping $\mathcal{L}(\fH)$ with its normalized Hilbert-Schmidt inner product.   There is a very useful construction that takes two orthonormal bases for $\fH$, and constructs from them an orthonormal basis for $\widehat{\fH}$. Expanding $\Phi\in \mathcal{L}(\fH)$ in bases constructed this way turns out to reveal significan properties of $\Phi$, as we now explain. 

It is convenient to index orthonormal bases of $\fH$ by ordered pairs 
$$(i,j) \in \{1,\dots n\}\times  \{1,\dots n\} =: \cJ_n\ .$$
 We use lower case greek letters to denote elements of the index set. For $\alpha = (i,j)\in \cJ_n$, $\alpha':= (j,i)$.

For $F,G\in M_n(\C)$, define the operator $\#(F\otimes G)$ on $\fH$ by
\begin{equation*}
\#(F\otimes G)X := FXG\ .
\end{equation*}
Simple computations show that 
\begin{equation*}
\langle \#(F_1\otimes G_1), \#(F_2\otimes G_2)\rangle_{\widehat{\fH}} = \langle F_1,F_2\rangle_{\fH} \langle G_1,G_2\rangle_{\fH}\ .
\end{equation*}
Therefore, whenever $\{F_\alpha\}_{\alpha\in \cJ_n}$   and $\{G_\alpha\}_{\alpha\in \cJ_n}$  are two orthonormal bases of $\fH$,  $\{\#(F_\alpha\otimes G_\beta)\}_{\alpha,\beta\in \cJ_n}$ is an orthonormal basis of $\widehat{\fH}$. 

Now fix any orthonormal basis $\{F_\alpha\}_{\alpha\in \cJ_n}$ of $\fH$.
Then $\{F^*_\alpha\}_{\alpha\in \cJ_n}$  is also an orthonormal basis of $\fH$, and hence $\{\#(F_\alpha^*\otimes F_\beta)\}_{\alpha,\beta\in \cJ_n}$ is an orthonormal basis of $\widehat{\fH}$.   Thus,
 every linear operator $\Phi$ on $\fH$ has an expansion
\begin{equation}\label{GKS1}
\Phi = \sum_{\alpha,\beta\in \cJ_n} (c_\Phi)_{\alpha,\beta}\#(F_\alpha^*\otimes F_\beta)\ ,
\end{equation}
where
\begin{equation}\label{GKS2}
(c_\Phi)_{\alpha,\beta} = \langle \#(F_\alpha^*\otimes F_\beta), \Phi \rangle_{\widehat{\fH}}\ .
\end{equation}
In particular, the coefficients $(c_\Phi)_{\alpha,\beta}$ are uniquely determined.

\begin{definition}[Characteristic matrix] Given $\Phi\in {\mathcal L}(\fH)$, and an orthonormal basis $\{F_\alpha\}_{\alpha\in \cJ_n}$ of $\fH$, 
its {\em characteristic matrix} for this orthonormal basis is the $n^2\times n^2$ matrix  $C_\Phi$ whose $(\alpha,\beta)$th entry is $[C_\Phi]_{\alpha,\beta}$ as specified in 
\eqref{GKS2}. 
\end{definition}

The following is a convenient formulation of a theorem of Choi \cite{Choi75}:

\begin{theorem}\label{CPchar}  A linear operator $\Phi$ on $M_n(\C)$  is completely positive if and only if for every orthonormal basis $\{F_\alpha\}_{\alpha\in \cJ_n}$ of $\fH$, the corresponding characteristic matrix $C_\Phi$ is positive semi-definite. 
\end{theorem}

\begin{proof} 
The right side of \eqref{GKS2} can be computed using the matrix unit basis $\{E_{k,\ell}\}$ to compute the normalized trace:
\begin{equation*}
[C_\Phi]_{\alpha,\beta} = \langle \#(F_\alpha^*\otimes F_\beta ), \Phi \rangle_{\widehat{\fH}} = \frac{1}{n^2} \sum_{1 \leq k,\ell\leq n} \tr[E_{k,\ell}^* F_\alpha\Phi(E_{k,\ell})F_\beta^*]\ .
\end{equation*} 
Let $(z_1,\dots,z_{n^2})\in \C^{n^2}$ and define $G := \sum_{\alpha\in \cJ_n} z_\alpha F_\alpha^*$. Then
\begin{equation}\label{GKSA1}
\sum_{\alpha,\beta }\overline{z}_\alpha (c_\Phi)_{\alpha,\beta} z_\beta   =  \frac{1}{n^2} \sum_{1 \leq k,\ell\leq n} \tr[E_{\ell,k} G^*\Phi(E_{k,\ell})G]\ .
\end{equation}
Let $[E_{i,j}]$ denote the block matrix whose $i,j$th entry is $E_{i,j}$. Then it is easy to see that $n^{-1/2}[E_{i,j}]$ is an orthogonal projection, and in particular, positive. 
Now suppose that $\Phi$ is completely positive. Then the block matrix $[G^*\Phi(E_{i,j})G]$ whose $i,j$th entry is $G^*\Phi(E_{i,j})G$ is positive. The right side of \eqref{GKSA1} is then
the trace (on the direct sum of $n$ copies of $\C^n$) of the product of positive $n^2\times n^2$ matrices, and as such it is positive.  Thus, whenever $\Phi$ is completely positive, $C_\Phi$ is positive semi-definite.

On the other hand, suppose that $C_\Phi$ is positive semi-definite. Let $\Lambda$ be a diagonal matrix whose diagonal entries are the eigenvalues of $C_\Phi$, and let $U$ be a unitary such that
$C_\Phi = U^*\Lambda U$.
Then by \eqref{GKS1}, for any $X\in \fH$,
$$
\Phi(X) =   \sum_{\alpha,\beta,\gamma\in \cJ_n}  U^*_{\alpha,\gamma} \lambda_k U_{\gamma,\beta }F_\alpha^*XF_\beta  =  \sum_{\gamma\in \cJ_n}  V_\gamma^* X V_\gamma \ ,
$$
where ${\displaystyle V_\gamma := \sqrt{\lambda_\gamma}\sum_{\alpha\in \cJ_n} U_{\gamma,\alpha}F_\alpha}$.
This shows that whenever $C_\Phi$ is positive semi-definite, $\Phi$ is completely positive, and provides a Kraus representation of it. 
\end{proof}

So far, we have not required any special properties of the orthonormal bases $\{F_\alpha\}_{\alpha\in \cJ_n}$ of $\fH$ that we used. To characterize QMS generators, it will be necessary to choose bases that have several useful properties:

\begin{definition}[symmetric, unital and matrix unit  bases]\label{unbadef}  
An orthonormal basis $\{F_\alpha\}_{\alpha\in \cJ_n}$ of $\fH$ is  {\em symmetric} in case
\begin{equation*}
{\rm For \ all}\quad \alpha\in \cJ_n\ ,\quad F_\alpha^* = F_{\alpha'}\ .
\end{equation*}
It is {\em unital}  in case it is symmetric and moreover
\begin{equation*}
F_{(1,1)} = \one\ .
\end{equation*}
It is the {\em matrix unit} basis corresponding to an orthonormal basis $\{u_1,\dots,u_n\}$ of $\C^n$ in case
\begin{equation*}
F_{(i,j)} =\sqrt{n}|u_i\rangle \langle u_j|\ .
\end{equation*}
Note that a matrix unit basis is symmetric. 
\end{definition}

One reason unital bases are useful is the following: Let us compute the characteristic matrix of the identity transformation for a unital basis $\{F_\alpha\}_{\alpha\in \cJ_n}$. Note that
$$\tr[F_\alpha] = \tr[F_{(1,1)}^* F_\alpha] = n \delta_{(1,1),\alpha}\ .$$
Therefore, proceeding as in the proof of Lemma~\ref{CPchar},
\begin{multline}\label{chid}
(C_I)_{\alpha,\beta} =   \frac{1}{n^2} \sum_{1 \leq k,\ell\leq n} \tr[E_{k,\ell}^* F_\alpha E_{k,\ell} F_\beta^*]  =\\ \frac{1}{n^2} \tr[F_\alpha]\tr[F_\beta] =   \delta_{\alpha,(1,1)} \delta_{\beta,(1,1)} \ .
\end{multline}
That is, $C_I$ is zero except in the upper left entry, where it has the value $1$. This fact will be useful below.

\begin{definition}[Reduced characteristic matrix]  Let $\cL$ be a Hermitian operator on $\fH$ such that $\cL\one = 0$, and let $\{F_\alpha\}_{\alpha\in \cJ_n}$ be a unital orthonormal basis. 
Let $C_\cL$ be the characteristic matrix of $\cL$ with respect to this basis.  The {\em reduced characteristic matrix} $R_\cL$ of $\cL$ is the $(n^2-1)\times (n^2-1)$ matrix obtained by deleting the first row and column of $C_{\cL}$.
\end{definition}

The following lemma is proved in \cite{GKS76}.

\begin{lemma}\label{redLm} Let $\cL$ be a Hermitian operator on $\fH$ such that $\cL\one = 0$, and let $\{F_\alpha\}_{\alpha\in \cJ_n}$ be a unital orthonormal basis. Let $R_\cL$ be the  reduced characteristic matrix of $\cL$ with respect to this basis.  Then $e^{t\cL} $ is completely positive for all $t>0$  if and only if  $R_\cL$ is positive semi-definite,
and in this case there is a completely positive operator $\Phi\in {\mathcal L}(\fH)$ and a $G\in M_n(\C)$ such that for all $A\in M_n(\C)$,
\begin{equation}\label{lindform}
\cL(A)  = \Phi(A) - G^*A - AG\ .
\end{equation}
\end{lemma}

\begin{proof} Suppose that $\cP_t := e^{t\cL}$ is completely positive for each  $t>0$. By \eqref{chid}, $R_I$, the reduced characteristic matrix of the identity, is $0$.
Then since
$$R_{t^{-1}(\cP_t - I)} = t^{-1}R_{\cP_t} -  t^{-1}R_I  =  t^{-1}R_{\cP_t}  \ ,$$
the reduced characteristic matrix of $t^{-1}(\cP_t - I)$ coincides with the reduced characteristic  matrix of 
$t^{-1}\cP_t $, and by Lemma~\ref{CPchar} this is positive. Taking the limit $t\to 0$, we conclude that the reduced characteristic matrix of $\cL$ is positive.

Conversely, suppose that the reduced characteristic matrix of $\cL$ is positive. 
Since $F_{(1,1)} = \one$, 
\begin{eqnarray*}
\cL(A)  = \sum_{\alpha,\beta}[C_\cL]_{\alpha,\beta} F_\alpha^* A F_\beta = -G^*A - AG  +   \sum_{\alpha,\beta}[R_\cL]_{\alpha,\beta} F_\alpha^* A F_\beta
\end{eqnarray*}
where 
$$G := -\frac12 (c_\cL)_{(1,1),(1,1)} \one - \sum_{\beta} (c_\cL)_{(1,1),\beta} F_\beta \ .
$$
By Theorem~\ref{CPchar}, if we define $\Phi(A) :=  \sum_{\alpha,\beta}[R_\cL]_{\alpha,\beta} F_\alpha^* A F_\beta$, then 
$\Phi$  is completely positive. This shows that for all $A\in M_n(\C)$,   $\cL(A)  = \Phi(A) - G^*A - AG$ whenever $R_{\cL}$ is positive semidefinite. Then by the  argument around \eqref{trotter}, 
$e^{t\cL}$ is completely positive for all $t> 0$, and $\cL$ has the form specified in \eqref{lindform}. 
\end{proof}

\begin{theorem}[Structure of QMS generators] Let $\cL$ be a QMS generator on $M_n(\C)$. Then there is a completely positive map $\Phi\in {\mathcal L}(M_n(\C))$ and a self adjoint 
$H\in M_n(\C)$ 
such that for all $A\ in M_n(\C)$, 
\begin{equation*}
\cL(A)  = \Phi(A) - \frac12( \Phi(\one)A + A\Phi(\one) + i[H,A] \ .
\end{equation*}
\end{theorem}

\begin{proof} By Lemma~\ref{redLm},  there is a completely positive operator $\Phi\in {\mathcal L}(\fH)$ and a $G\in M_n(\C)$ such that for all $A\in M_n(\C)$,
$\cL(A)  = \Phi(A) - G^*A - AG$. Since $\cL$ is a QMS generator, $\cL\one =0$, and hence $\Phi(\one) = G^* + G$.  Therefore, $G = \frac12 \Psi(\one) + i H$ for some self 
adjoint $H\in M_n(\C)$.
Then
$$
G^*A + AG  = \frac12\Phi(\one) A + A \frac12\Phi(\one) -iHA + i AH  =  \frac12\Phi(\one) A + A \frac12\Phi(\one) -i[H,A]\ .
$$
\end{proof} 

\begin{remark} We have essentially followed the proof of \cite{GKS76} with some simplifications from \cite{PZ77}. The proof of Lindblad \cite{Lin76} is somewhat different and has the advantage of extending to some infinite dimensional cases.
\end{remark}

\subsection{Structure of generators satisfying detailed balance}

A classical Markov chain on the state space $\{1,\dots,n\}$ with transition matrix $P$ and invariant measure $\mu$ satisfies the {\em detailed balance condition} in case for all $1 \leq i,j \leq n$
\begin{equation}\label{DBCC}
\mu(i)P_{i,j} = \mu(j)P_{j,i}
\end{equation}
which means that in the steady state, a transition from $i$ to $j$ has exactly the same probability as a transition from $j$ to $i$. If one is watching a film of the process, it is impossible to know if the film is running backwards or forwards. For this reason, the detailed balance condition is sometimes referred to as {\em microscopic reversibility}.   If we equip the space of functions  on the state space  $\{1,\dots,n\}$ with the inner product $\langle \cdot,\cdot\rangle_\mu$ defined by
$$
\langle f,g\rangle_\mu = \sum_{j=1}^n  \mu(j) \overline{f(j)}g(j)\ ,
$$
and define $Pf(i) = \sum_{j=1}^nP_{i,j}f(j)$, then given \eqref{DBCC}
$$
\langle f,Pg\rangle_\mu = \sum_{i,j=1}^n \mu(i)P_{i,j} \overline{f(i)}g(j) = \sum_{i,j=1}^n \mu(j)P_{i,j} \overline{f(i)}g(j) = \langle Pf,g\rangle_\mu\ .
$$
Thus, \eqref{DBCC} implies that $P$ is self adjoint with respect to the inner product $\langle \cdot,\cdot\rangle_\mu$, and a little thought shows that the argument reverses to that in fact the self adjointness of $P$ with respect to  inner product $\langle \cdot,\cdot\rangle_\mu$ characterizes the detailed balance condition. 

This simple observation opens the way to defining the detailed balance condition for QMS semigroups $\cP_t$ with invariant state $\sigma$. Actually, it opens many ways, since, as we have  seen, there are many choices for the inner product associated to $\sigma\in \Dens$. 

There is an extensive discussion of quantum detailed  balance in the the physics literature, and in particular on how it related to reversibility.  For early papers from the 1970's,
see
\cite{ A78, Ar73, CW76,KFGV,SpLe77}. In these papers various different proposals are discussed, and recent work \cite{FR15, FU07,FU10,  MS98, TKRWV} has added to the variety of 
proposals. Indeed, the ``right'' choice of the inner product depends on the physical question being asked, and as emphasized already in the early work of Agarwal \cite{Ar73}, the connection with reversibility may involve other operations. In any case, two choices that have been considered to be particularly relevant in various physical contexts are the GNS and KMS inner products. 

As will be seen here, the following choice is natural for our context:

\begin{definition}\label{DBCD}
The generator $\cL$ of a QMS   satisfies {\em detailed balance condition with respect to $\sigma$} in case $\cL$ is self-adjoint with 
regard the the GNS inner product associated to $\sigma$. 
\end{definition}

\begin{remark} Whenever  $\cL$  satisfies  detailed balance condition with respect to $\sigma$, $\cL^\dagger\sigma =0$ since for any $A$, 
$$
0 = \langle A,\cL\one\rangle_{{\text GNS}} = \langle \cL A,\one\rangle_{{\text GNS}}  = \tr[(\cL A)^* \sigma] = \tr[A^* \cL^\dagger \sigma]\ .$$

\end{remark}

Note that the GNS inner product is the inner product determined by $m_0 = \delta_0$. 
The key fact that makes an explicit structure theorem possible in this case is due to Alicki \cite{Al78}:

\begin{lemma}\label{Al2}  Let $\sigma \in \Dens$ be a non-degenerate density matrix, and let $s\in [0,1]$, $s\neq 1/2$. 
Let $\cK$ be any operator on $M_n(\C)$ that is self-adjoint with respect to $\langle \cdot, \cdot\rangle_{\delta_s}$ and also preserves self-adjointness. 
Then $\cK$ commutes with  the modular operator $\Delta_\sigma$. 
\end{lemma}

\begin{proof}
For any $A,B\in M_n(\C)$,
\begin{eqnarray*}
\langle  \cK \Delta_\sigma^{2s-1}(A), B\rangle_{\delta_s} 
&=& \tr[ \sigma^s (\cK (\sigma^{2s-1}A \sigma^{1-2s}))^* \sigma^{1-s} B]\nonumber \\
&=& \tr[ \sigma^s (\sigma^{2s-1}A \sigma^{1-2s})^*  \sigma^{1-s}\cK ( B)]\nonumber\\
&=& \tr[  \sigma^{1-s}A^* \sigma^{s}  \cK ( B)]
= \tr[\sigma^{s}  (\cK ( B^*))^*  \sigma^{1-s}A^* ]\nonumber\\
&=& \tr[\sigma^{s}  B  \sigma^{1-s}\cK (A^*) ] = \tr[\sigma^{1-s}( \cK(A))^* \sigma^s B]  = \langle \cK (A), B\rangle_{\delta_{1-s}}\ .
\end{eqnarray*}

By a simple computation, for all $s,t\in [0,1]$,
\begin{align*}
\langle \Delta_\sigma^t (A), B\rangle_{\delta_s} 
= \tr[ \sigma^{s}(\sigma^t A \sigma^{-t})^* \sigma^{1-s}B]  
= \tr[ \sigma^{s-t} A^* \sigma^{1-s+t}B] 
= \langle A,B\rangle_{\delta_{s-t} }
\ .	\end{align*}
Since $s- (2s-1) = 1-s$, $\langle \cK (A), B\rangle_{\delta_{1-s}} = \langle  
\Delta_\sigma^{2s-1}(\cK (A)), B\rangle_{\delta_s}$. 
As $B$ is arbitrary,  $\Delta_\sigma^{2s-1} \cK  = \cK \Delta_\sigma^{2s-1}$. 
Since $\cK $ commutes with $\Delta_\sigma^{2s-1}$, it commutes with 
$ f(\Delta_\sigma^{2s-1})$ for every function $f$. 
\end{proof} 

In particular, if $\cL$ satisfies detailed balance, then
\begin{equation*}
\cL A =  \sigma^{-1}( \cL (\sigma A\sigma^{-1})) \sigma\ .
\end{equation*}

We now consider a unital orthonormal basis for $\fH$ consisting of eigenvectors of $\Delta_\sigma$. 
This is possible since $\Delta_\sigma$ is a positive operator.   The eigenvectors are of the form $|u_i\rangle\langle u_j|$ where 
$$\sigma u_j = \lambda _j u_j$$
and then 
$$
\Delta_\sigma(|u_i\rangle\langle u_j|) = e^{\omega_{(i,j)}}|u_i\rangle\langle u_j|\quad{\rm where}\quad \omega_{i,j} :=  \log \lambda_i - \log \lambda_j\ .
$$
Computing the characteristic matrix $C_{\alpha,\beta}$ for such a basis we find
\begin{eqnarray}
 \sum_{\alpha,\beta} C_{\alpha,\beta} F_\alpha^*A F_\beta &=& \cL A 
=  \sigma^{-1}( \cL (\sigma A\sigma^{-1}))\sigma  =   
\sum_{\alpha,\beta} 
C_{\alpha,\beta} \sigma^{-1}F_\alpha^*\sigma A \sigma^{-1}F_\beta\sigma \nonumber\\
&=&  \sum_{\alpha,\beta} C_{\alpha,\beta} 
e^{\omega_\beta - \omega_\alpha} F_\alpha^*A F_\beta\ .\nonumber
\end{eqnarray}
In other words, 
\begin{equation}\label{block}
e^{\omega_\alpha}C_{\alpha,\beta} = C_{\alpha,\beta}e^{\omega_\beta} \ .
\end{equation}
Since the characteristc matrix is self adjoint, (\ref{block}) implies that
\begin{equation*}
\omega_\alpha \neq \omega_\beta\quad  \Rightarrow \quad C_{\alpha,\beta} = 0\ .
\end{equation*}
Therefore with an ordering of the indices $\alpha$ so that 
$\alpha \geq \beta \iff \omega_\alpha \geq \omega_\beta$, the matrix 
$C_{\alpha,\beta}$ is block-diagonal.   Now one need only recall how the Kraus operator in the representation of $\cL$ are determined by the eigenvectors of the reduced density matrix of $\cL$, and one has Alicki's Thoerem \cite{Al78}:

\begin{theorem}[Structure of Lindblad operators with detailed balance] \label{thm:structure}
Let $\cP_t = e^{t \cL}$ be a quantum Markov semigroup on $B(\sH)$ satisfying detailed balance with respect to $\sigma \in \Dens$.
Then the generator $\cL$ and its adjoint $\cL^\dagger$ have the form
\begin{align}\label{eq:L-general}
	\cL & = \sum_{j \in \cJ}  e^{-\omega_j/2} \cL_j \ , \qquad
	\cL_{j}(A) =  V_j^* [A, V_j] +  [V_j^* , A] V_j \ , \\
\label{eq:L-dagger-general}
	\cL^\dagger & = \sum_{j \in \cJ}  e^{-\omega_j/2} \cL_{j}^\dagger \ , \qquad 
		\cL^\dagger_{j}(\rho) = 
		  [V_j, \rho V_j^*] + [V_j \rho, V_j^* ] \ ,
\end{align}
where $\cJ$ is a finite index set, the operators $V_j \in B(\sH)$ satisfy $\{ V_{j} \}_{j \in \cJ} = \{ V_{j}^{*} \}_{j \in \cJ}$, and $\omega_{j} \in \R$ satisfies
\begin{align}
\label{eq:log-relation} 
\Delta_\sigma V_j & = e^{-\omega_j} V_j \quad \text{for all } j \in \cJ \ .
\end{align}
\end{theorem}

For $j \in \cJ$, let $j^{*} \in \cJ$ be an index such that $V_{j^{*}} = V_{j}^{*}$. It follows from \eqref{eq:log-relation} that 
\begin{align*}
\omega_{j^{*}} = - \omega_{j} \ .
\end{align*}
Moreover, if we define $H = -\log \sigma$, \eqref{eq:log-relation} is equivalent to the commutator identity 
$
[V_j,H]  = - \omega_{j}V_j
$.
Furthermore, in our finite-dimensional context, the identity 
\begin{align*}
\Delta_\sigma^t V_j  = e^{- \omega_j t} V_j
\end{align*}
is valid for some $t\neq 0$ in $\R$ if and only if it is valid for all $t \in \C$. 

There is also a well developed structure theorem for generators of QMS semigroups that are self-adjoint with respect to the KMS inner product; see \cite{AC20,FU10}.  Unfortunately, 
there does not yet exist such a theorem for the set of QMS generators that are self-adjoint with respect to the BKM inner product, which strictly includes \cite{AC20} the set of those that 
are self-adjoint for the GNS inner product considered here. By Theorem~\ref{necBKM}, this would  be the most natural class of QMS to consider.  However, the metrics to be constructed in the 
remainder of this lecture make use of the structure theorem, Theorem~\ref{thm:structure}, and its analog for the BKM case is presently lacking. 

\subsection{Gradient flow for the relative entropy}

Consider a generator $\cL^\dagger$ written in the form \eqref{eq:L-dagger-general}, i.e., 
\begin{align}\label{Ldet}
	\cL^\dagger & = \sum_{j \in \cJ}  e^{-\omega_j/2} \cL_{j}^\dagger \ , \qquad 
		\cL^\dagger_{j}(\rho) = 
		  [V_j, \rho V_j^*] + [V_j \rho, V_j^* ] \ ,
\end{align}
where $\{V_j\}_{j\in \cJ}$ is a finite set of eigenvectors of $\Delta_\sigma$ such that $\{V_j^*\}_{j\in \cJ} = \{V_j\}_{j\in \cJ}$, and where $\Delta_\sigma V_j = e^{-\omega_j}V_j$ for some $\omega_j \in \R$. 
As before, we use the notation $\partial_j A := [V_j, A]$.

The next lemma will be a chain rule for the logarithm. It gives us a non-commuative analog and extension  of the formula
$$
\nabla \rho(x) = \rho(x) \nabla \log\rho(x)\ ,
$$
which behind the fact that one recovers the usual linear heat equation from the gradient flow of the entropy with respect to the $2$-Wasserstein metric. 
The relevant non-commutative analog of multiplication by $\rho$ is somewhat involved, but is also something we have discussed.
Define
\begin{equation*}
f_0(x) = \int_0^1 x^{s}{\rm d}s = \frac{1- x}{\log x}
\end{equation*}
on $(0,\infty)$. 

Recall the operator introduced by Hiai and Petz:
\begin{equation*}
{\mathbb J}_f(X,Y) := f(R_X L_Y^{+} )L_Y ,
\end{equation*}
for any $Y,X\geq 0$ and any operator monotone function $f$ on $(0,\infty)$.
Taking $f$ to be $f_0$, we find
$$
{\mathbb J}_{f_0}(X,Y)(A) = \int_0^1  Y^{1-s}A X^{s} {\rm d} s\ .
$$ 
Our formulas will involve
\begin{equation}\label{Jfomf0}
{\mathbb J}_{f_0}(e^{\omega_j/2}\rho,e^{-\omega_j/2}\rho)(A) = \int_0^1 \big(e^{\omega_j /2} \rho\big)^{1-s}
   			 A  \big(e^{-\omega_j /2} \rho\big)^{s}\dd s\ .
\end{equation}
Evidently, if $A$ commutes with $\rho$,
$$
{\mathbb J}_{f_0}(e^{\omega_j/2}\rho,e^{-\omega_j/2}\rho)(A) = e^{\omega_j/2}\rho A\ .
$$
That is, ${\mathbb J}_f(e^{\omega_j/2}\rho,e^{-\omega_j/2}\rho)$ is nothing other than a  non-commutative version of multiplication by $e^{\omega_j/2}\rho$.
Its inverse, which we have encountered in the first lecture, is
$$
{\mathbb J}_{f_0}^{-1}(e^{\omega_j/2}\rho,e^{-\omega_j/2}\rho) (A) =
\int_0^\infty
	    \frac{1}{\lambda + e^{\omega/2}\rho} A
\frac{1}{\lambda + e^{-\omega/2} \rho}  \dd \lambda\ .
$$

We are now ready to state our fundamental lemma \cite{CM17} which may be regarded as a non-commutative chain rule for the logarithm.

\begin{lemma}[Chain rule for the logarithm]\label{chain}
For all $\rho \in \Dens_+$ and $j \in \cJ$ we have
	\begin{align*}
		 e^{-\omega_j/2}V_j \rho  - e^{\omega_j/2}\rho V_j
		 &  =   \left[ {\mathbb J}_{f_0}(e^{\omega_j/2}\rho,e^{-\omega_j/2}\rho)\right] \partial_j(\log \rho - \log \sigma) \ .
	\end{align*}
\end{lemma}

\begin{proof}
	Since $\Delta_\sigma V_j  = e^{-\omega_j} V_j$ for all $j$,  $\sigma^s V_j\sigma^{-s} = \Delta^s_\sigma V_j  = e^{-s\omega_j} V_j$. Differentitating at $s=0$,
	$$
	[V_j,\log\sigma] = \omega_j V_j\ .
	$$
	\begin{align*}
		\partial_{j}( \log \rho - \log \sigma  ) 
		   &= [V_j,\log \rho] - [V_j,\log \sigma]\\
		   &= V_j \log\rho - \log\rho V_j  -\omega_jV_j    \\
		   &= V_j \log (e^{-\omega_j/2} \rho) - \log (e^{\omega_j/2} \rho) V_j \ .
	\end{align*}
	Therefore, 
	\begin{multline*}
	\int_0^1 \big(e^{\omega_j /2} \rho\big)^{1-s}
   			 [\partial_{j}( \log \rho - \log \sigma  ) ]  \big(e^{-\omega_j /2} \rho\big)^{s}\dd s  =\\
		\int_0^1 \big(e^{\omega_j /2} \rho\big)^{1-s}
   			 [ V_j \log (e^{-\omega_j/2} \rho) - \log (e^{\omega_j/2} \rho) V_j ]  \big(e^{-\omega_j /2} \rho\big)^{s}\dd s\ .	 
	\end{multline*}
	However,
	$$
	\big(e^{\omega_j /2} \rho\big)^{1-s}
   			  [V_j \log (e^{-\omega_j/2} \rho) - \log (e^{\omega_j/2} \rho) V_j]  \big(e^{-\omega_j /2} \rho)^{s}  =
			   \frac{{\rm d}}{{\rm d}s} 
			   (e^{\omega_j /2} \rho)^{1-s} V_j 
			 (e^{\omega_j/2} \rho)^{s}\ .		 
	$$
	The claim now follows from \eqref{Jfomf0} and the integration of the total derivative. 
\end{proof}

For $\rho \in \Dens$ we define the operator $\cK_\rho : {\mathcal V} \to {\mathcal V}$ by
\begin{align}\label{eq:K-rho}
	\cK_\rho A := \sum_{j \in \cJ} \partial_{j}^\dagger \big({\mathbb J}_{f_0}(e^{\omega_j/2}\rho,e^{-\omega_j/2}\rho) ( \partial_{j} A) \big) \ .
\end{align}
Since $\tr(A^* \cK_\rho B) = \overline{\tr(B^* \cK_\rho A)}$ for $A, B \in {\mathcal V}$, it follows that $\cK_\rho$ is a non-negative self-adjoint operator on $\fH$ for each $\rho \in \Dens_+$. 
Recall that ${\mathcal V}$ is the traceless self adjoint matrices in $M_n(\C)$ which we have identified with the tangent and cotangent spaces to $\Dens$. 
Assuming that $\cP_t$ is ergodic, the operator $\cK_\rho : {\mathcal V} \to {\mathcal V}$ is invertible for each $\rho \in \Dens_+$, as we explain below. Since $\cK_\rho$ depends smoothly on $\rho$, it follows that $\cK_\rho$ induces a Riemannian metric on $\Dens_+$ defined by \eqref{eq:metric-K}.

\begin{definition}\label{metricdef} Let $\cL$ be a QMS generator satisfying the detailed balance condition with respect to $\sigma\in \Dens$.  For each $\rho\in \Dens$, let the operator $\cK_\rho$ be defined by \eqref{eq:K-rho}.  Then, provided $\cK_\rho$ is invertible for all $\rho\in \Dens$,  {\em the Riemannian metric on $\Dens$ induced by $\cL$} is given by
\begin{equation}\label{grhodefag}
g_\rho(\dot \rho,\dot\rho) = \langle \dot \rho, \cK_\rho^{-1}\dot \rho\rangle = \langle (\cK_\rho^{-1}\dot \rho), \cK_\rho(\cK_\rho^{-1}\dot \rho)\rangle  \ .
\end{equation}
\end{definition}

The following result \cite{CM17} shows that the Kolmogorov forward equation $\partial_t \rho = \cL^\dagger \rho$ can be formulated as the gradient flow equation for $\Ent_{\sigma}$. A different construction of such a metric was given at the same time by Mielke and Mittenzweig \cite{MM17}

\begin{proposition}\label{prop:Lindblad-grad-flow}
For $\rho \in \Dens_+$ we have the identity
	\begin{align*}
		 \cL^\dagger \rho = - \cK_\rho \rmD D(\rho ||\sigma)\ ,
	\end{align*}
	hence the gradient flow equation of $\Ent_\sigma$ with respect to the Riemannian metric induced by $\cK_\rho$ is given by the solution of the  equation $\partial_t \rho = \cL^\dagger \rho$. 
\end{proposition}

\begin{proof}
	Using the  chain rule from Lemma~\ref{chain}, and the fact that $\{V_j\} = \{V_j^*\}$ and $\omega_{j^*} = -\omega_j$, we obtain 
	\begin{align*}
		\cK_\rho \rmD D(\rho||\sigma)
		 & =  \sum_{j \in \cJ} 
 	\partial_{j}^\dagger \big( {\mathbb J}_{f_0}(e^{\omega_j/2}\rho,e^{-\omega_j/2}\rho) \partial_{j} (\log \rho - \log \sigma) \big)  
		\\& =   \sum_{j \in \cJ} 
 	\partial_{j}^\dagger \big(  e^{-\omega_j/2}V_j \rho  - e^{\omega_j/2}\rho V_j \big)   
		\\&  = \frac12 \sum_{j \in \cJ} \Big(
 	\partial_{j}^\dagger \big(  e^{-\omega_j/2}V_j \rho  - e^{\omega_j/2}\rho V_j \big)   
 	    + 
 	\partial_{j} \big(  e^{\omega_{j}/2}V_j^* \rho  - e^{-\omega_j/2}\rho V_j^* \big)   \Big)
		 \\& = - \frac12\sum_{j \in \cJ}
		 e^{-\omega_j/2} \Big( [V_j, \rho V_j^*] +  [V_j \rho, V_j^* ] \Big) 
		+ e^{\omega_j/2} \Big(  [V_j^*, \rho V_j] +  [V_j^* \rho, V_j ]\Big) 
		 \\& = - \sum_{j \in \cJ}
		 e^{-\omega_j/2} \Big( [V_j, \rho V_j^*] +  [V_j \rho, V_j^* ] \Big) 
		  = - \cL^\dagger \rho \ ,
	\end{align*}
which is the desired identity.
\end{proof}

The gradient flow structure given in Proposition \ref{prop:Lindblad-grad-flow} is a non-commutative analogue of the Kantorovich gradient flow structure established by Jordan, Kinderlehrer and Otto \cite{JKO} for the \emph{Kolmogorov forward equation}
\begin{align*}
\frac{\partial}{\partial t}\rho(x,t) = \Delta \rho(x,t) - \nabla \cdot ( \rho(x,t) \nabla \log \sigma(x) ) \ .
\end{align*} 
Their gradient flow structure is formally given in terms of the operator $K_\rho$ defined by 
\begin{align*}
 K_{\rho} \psi  = - \nabla\cdot(\rho \nabla \psi) \ ,
\end{align*}
for probability densities $\rho$ on $\R^n$ and suitable functions $\psi : \R^n \to \R$ in analogy with \eqref{eq:K-rho}.
As the differential of the relative entropy $\Ent_\sigma(\rho) = \int_{\R^n} \rho(x)\log \frac{\rho(x)}{\sigma(x)} \dd x$ is given by $\rmD \Ent_\sigma(\rho) = 1 + \log \frac{\rho}{\sigma}$, we have
\begin{align*}
   K_\rho \rmD \Ent_\sigma(\rho) = - \Delta \rho + \nabla\cdot (\rho \nabla \log \sigma)\ ,
\end{align*}
which is the commutative counterpart of Proposition \ref{prop:Lindblad-grad-flow}. 

\subsection{Non commutative differential structure}

To bring out the analogy made at the end of he previous section, we introduce a non-commmutative gradient and divergence associated to $\cL$, as specified by $\sigma$ and 
$\{V_1,\dots,V_{|\cJ|}\}$.

Define the Hilbert space $\fH_{\cJ}$ by
$$\fH_{\cJ} = \bigoplus_{j\in \cJ} \fH_j \ ,$$
where each $\fH_j$ is a copy of $\fH$. 
For ${\bf A}\in \fH_{\cJ}$ and $j \in \cJ$, let $A_j$ denote the component of ${\bf A}$ in $\fH_j$. Thus, picking some linear ordering of $\cJ$, we can write
$${\bf A} = ( A_1, \dots, A_{|\cJ|})\ .$$
We equip $\fH_{\cJ}$ with the usual inner product
$
\langle {\bf A}, {\bf B}\rangle_{\fH_{\cJ}} = \sum_{j\in \cJ}\langle A_j,B_j\rangle_{\fH_j}$.

Define an operator
$\nabla : \fH \to \fH_{\cJ}$ by
$$\nabla A = ( \partial_1 A, \dots, \partial_{|\cJ|}A)\ ,$$
where, as before, $\partial_j A = [V_j,A]$. 
Thinking of elements of $\fH$ as non-commutative analogs of functions on a manifold, we may think of
${\bf A} = ( A_1, \dots, A_{|\cJ|})$ as a vector field. 
We define the operator $\dive : \fH_{\cJ}\to \fH$ by
$$\dive {\bf A} 
  = - \sum_{j\in \cJ}  \partial_j^\dagger A_j
  = \sum_{j\in \cJ}   [A_j,V_j^*] \ .$$
Note that $\dive$ is minus the adjoint of the map $\nabla: \fH \to  \fH_{\cJ}$, so that 
\begin{equation}\label{cL0def}
\cL_0  := \dive \circ \nabla
\end{equation} is negative semi-definite.
 We call $\nabla$ the {\em non-commutative gradient} associated to $\cL$, and $\dive$ the {\em non-commutative divergence} associated to $\cL$.

Now for a vector ${\bf A} =  ( A_1, \dots, A_{|\cJ|})$, we define
\begin{equation}\label{Mudef}
{\mathbb M}_\rho {\bf A} = ({\mathbb J}_{f_0}(e^{\omega_1/2}\rho,e^{-\omega_1/2}\rho)(A_1),\dots, 
{\mathbb J}_{f_0}(e^{\omega_{|\cJ|}/2}\rho,e^{-\omega_{|\cJ|}/2}\rho)(A_{|\cJ|}))\ .
\end{equation}

Now the operator $\cK_\rho$ can be writen as
\begin{equation}\label{Mudef2}
\cK_\rho(A) = -\dive({\mathbb M}_\rho \nabla A) \ .
\end{equation}
The following was proved in \cite{CM17}.

\begin{theorem}\label{erg} 
Let $\cP_t = e^{t\cL}$ be QMS  that satisfies the detailed balance condition with respect to  $\sigma\in \Dens$.  
Let $\cL$ be given in the form \eqref{Ldet}. 
Then the commutant of $\{V_j\}_{j\in \cJ}$ equals the null space of $\cL$.
In particular, $\cP_t$ is ergodic if and only if the commutant of $\{V_j\}_{j\in \cJ}$ is spanned by the identity.   
\end{theorem} 

\begin{proof} 
Suppose that $A$ belongs to the commutant of $\{V_j\}_{j\in \cJ}$. 
By definition, this means that $\partial_j A = 0$ for all $j \in \cJ$, and therefore $A \in {\rm Null}(\cL)$ by \eqref{Ldet}.

Conversely, if $\cL A = 0$, then by \eqref{Ldet} and a simple computation,
$$
0   = - \tr[ A^* \sigma^{1/2}\cL A \sigma^{1/2}]
	= \sum_{j\in \cJ}  
		\tr[ (\partial_{j} A)^*\sigma^{1/2}\partial_{j} A \sigma^{1/2}] \ , 
$$
which is the case if and only if $[V_j,A] = 0$ for all $j$. This means that $A$ belongs to the commutant of  $\{V_j\}_{j\in \cJ}$. 
\end{proof}

\begin{definition}[Gradient subspace]\label{gradsubdef}
Define ${\mathfrak G}_{\cJ}$ to be the subspace of $\fH_{\cJ}$ consisting of gradients:
\begin{equation*}
{\mathfrak G}_{\cJ} := \{ \nabla X \ : \ X\in M_n(\C)\ \}\ .
\end{equation*}
\end{definition}

Theorem~\ref{erg} shows that  if the commutant of $\{V_j\}_{j\in \cJ}$ is spanned by the identity and ${\mathcal V}$ denotes, as before, the space of self adjoint traceless elements of $M_n(\C)$,
then $\nabla$ is invertible from  ${\mathcal V}$  to ${\mathfrak G}_{\cJ}$.  In fact, it is easy to write down the explicit inverse. Under our hypotheses, $\cL_0$ is invertible as a map from 
${\mathcal V}$ to ${\mathcal V}$.  Let $\cL_0^+$ denote the generalized inverse of $\cL_0$ on $M_n(\C)$. Then
${\displaystyle
\left(\nabla\big|_{{\mathcal V}}\right)^{-1} = \cL_0^+\cdot\dive}$
This map will prove useful in the next lecture.

Another obvious consequence of Theorem~\ref{erg}  is that $\cK_\rho$ is positive definite on ${\mathcal V}$, the space of self adjoint traceleless elements of $M_n(\C)$. 
 Thus, $\cK_\rho$ defines a bona-fide Riemannian metric on $\Dens$, and it is a close analog of the classical version 
 $K_{\rho} \psi  = - \nabla\cdot(\rho \nabla \psi)$.
 
The following theorem summarizes a number of the results presented in this lecture, and provides the basis for the applications to be made in the next lecture.

\begin{theorem}\label{ergX2} 
Let $\cP_t = e^{t\cL}$ be QMS  that satisfies the detailed balance condition with respect to for $\sigma\in \Dens$.  
Let $\cL$ be given in the form \eqref{Ldet}.  Suppose that
 the commutant of $\{V_j\}_{j\in \cJ}$ is spanned by the identity.  Define a metric $g_\rho$ on $\Dens$ by \eqref{grhodefag}.
 Then the solution of the equation ${\displaystyle \partial_t \rho = \cL^\dagger \rho}$ is given by the gradient flow for the relative entropy functional $D(\rho||\sigma)$ with respect to the Riemannian metric specified in Definition~\ref{metricdef}.
\end{theorem}

In the next lecture, we explain why this metric is naturally viewed as a quantum, or non-commutative, mass transport metric.  
The subject of quantum mass transport metrics has been actively developed in recent years by many researchers with many 
different applications in mind.  The works \cite{CGT18,CGTT20,CLZ,DR20,GR21,GR22,GJL20,MM17,MW18} all 
 investigate quantum mass transportation from a dynamical perspective involving quantum Markov equations, and contain much of interest, although the present discusison sticks close to the approach of \cite{CM14,CM17,CM20}.  There are also a number of works dealing with quantum mass transport from rather different perspectives. Among these are \cite{CGP21,DMTS21,DR22,DT21,GP18,RD19}. It would be interesting to clarify the connections between these various works.

%% file: chapterC3.tex
%
%
%
\section{ Applications of the Gradient Flow Structure}


\abstract{\footnotesize There are many motivations for writing Quantum Markov Semigroups in terms of gradient flow. One is that this provides a powerful means of proving certain useful functional inequalities, a point of view pioneered in the classical setting by Felix Otto.  This approach is very geometric and curvature bounds for the metrics constructed in the second lecture play a fundamental role. In the third lecture we shall discuss a number of applications and also a number of open problems.}

\subsection{Geodesic convexity and relaxation to equilibrium}

In this section we develop the advantages of having written the evolution equations as gradient flow for the relative entropy. 
We draw on work of Otto and Westdickenberg \cite{OW05} and Daneri and Savar\'e \cite{DS08}. 

Let $(\cM,g)$ be any smooth, finite-dimensional Riemannian manifold.  
For $x,y$ in $\cM$, the Riemannian distance $d_g(x,y)$ between $x$ and $y$ is given by minimizing an action integral of paths $\gamma:[0,1]\to \cM$ running from $x$ to $y$:
\begin{equation*}
d_g^2(x,y) = \inf\left\{ \int_0^1 \|\dot{\gamma}(s)\|^2_{g(\gamma(s))} 
\dd s\ :\ \gamma(0) =x,\ \gamma(1) = y\right\} \ ,
\end{equation*}
where
\begin{equation*}
\|\dot \gamma(s)\|^2_{g(\gamma(s))} =  g_{\gamma(s)}(\dot \gamma(s), \dot \gamma(s)) \ .
\end{equation*}
(If the infimum is achieved, any minimizer $\gamma$ will be a geodesic.)  
If $F$ is a smooth function on $\cM$,
let ${\rm grad}_g F$ denote its Riemannian gradient. Consider the semigroup $S_t$ 
of transformations on $\cM$
given by solving $\dot \gamma(t) = -{\rm grad}_g F(\gamma(t))$; we assume  that nice global solutions exist, which will be the case in our application. The semigroup $S_t$, $t\geq 0$, is  gradient flow for $F$. 

For $\lambda \in \R$, the function $F$ is $\lambda$-convex  in case whenever $\gamma:[0,1]\to \cM$ is a distance minimizing geodesic, then for all $s\in (0,1)$,
\begin{equation*}
\frac{{\rm d}^2}{{\rm d}s^2} F(\gamma(s)) \geq \lambda g(\dot \gamma(s), \dot \gamma(s))\ .
\end{equation*}

It is a standard result that whenever $F$ is $\lambda$-convex, the gradient flow for $F$ is $\lambda$-contracting in the sense that for all $x,y\in \cM$ and $t > 0$,
\begin{equation*}
\frac{{\rm d}}{{\rm d}t} d_g^2(S_t(x),S_t(y))  \leq -2\lambda d_g^2(S_t(x),S_t(y)) \ .
\end{equation*}

We apply these ideas  here using an approach to geodesic convexity developed  by Otto and Westdickenberg \cite{OW05}. 
Let $\{\gamma(s)\}_{s\in [0,1]}$ be any smooth path in $\cM$ with $\gamma(0)= x$ and $\gamma(1) = y$.  
They use the gradient flow transformation $S_t$ to define a one-parameter family of paths $\gamma^t:[0,1]\to \cM$, $t\geq 0$ defined by
\begin{equation*}
\gamma^t(s) = S_t\gamma(s)\ .
\end{equation*}
Since $\gamma^t$ is admissible for the variational problem that defines $d_g(S_t(x),S_t(y))$, it is immediate that for each $t\geq 0$, 
\begin{equation}\label{act5}
d_g^2(S_t(x),S_t(y)) \leq \int_0^1 \left\| \frac{{\rm d}}{{\rm d}s}\gamma^t(s)\right\|^2_{g(\gamma^t(s))} \dd s \ .
\end{equation}
In the present smooth setting it is shown in \cite[(2.8) -- (2.11)]{DS08}  that if for all smooth curves
$\gamma:[0,1]\to\cM$, 
\begin{equation}\label{act6}
\frac{{\rm d}}{{\rm d}t}\bigg|_{0+}  \left( \left\| \frac{{\rm d}}{{\rm d}s}\gamma^t(s)\right\|^2_{g(\gamma^t(s))}\right)
\leq -2\lambda 
 \left\| \frac{{\rm d}}{{\rm d}s}\gamma^0(s)\right\|^2_{g(\gamma^0(s))}\ ,
\end{equation}
for all $s\in (0,1)$, then $F$ is geodesically $\lambda$-convex.

To see the connection between (\ref{act6}) and the contraction property, suppose that $x$ and $y$ are connected by a minimal geodesic $\gamma$ so that 
$$d_g^2(x,y) = \int_0^1 \left\| \frac{{\rm d}}{{\rm d}s}\gamma(s)\right\|^2_{g(\gamma(s))} \dd s\ .$$
(If $x$ and $y$ are sufficiently close, this is the case.)
Then (\ref{act5}) and (\ref{act6}) combine to yield 
\begin{equation*}
\frac{{\rm d}}{{\rm d}t}\bigg|_{0+}   d_g^2(S_t(x), S_t(y))  \leq -2\lambda 
d_g^2(x,y)\ ,
\end{equation*}
and then, provided that $S_t(x)$ and $S_t(y)$ continue to be connected by a minimal geodesic for all $t$, the semigroup property of $S_t$ yields the exponential $\lambda$-contractivity of the flow:
\begin{equation}\label{act8}
d_g(S_t(x),S_t(y))   \leq e^{-\lambda t}d_g(x,y)\ .
\end{equation}
The local argument in \cite{DS08} proves the geodesic $\lambda$-convexity of $F$ when (\ref{act6}) is valid for all smooth paths in $\cM$, and thus leads to (\ref{act8}) without any assumptions of geodesic completeness.

When  $F$ is  $\lambda$-convex  for some $\lambda>0$, so that  (\ref{act6}) and hence also (\ref{act8}) are satisfied, then, $F$ has at most one fixed point $x_0$ in $\cM$, which is necessarily a strict minimizer of $F$ on $\cM$.  
Subtract a constant to obtain  $F(x_0) = 0$. Then under the geodesic $\lambda$-convexity of $F$, one has
\begin{equation}\label{act9}
\frac{{\rm d}}{{\rm d}t} F(S_t(x)) \leq - 2\lambda F(S_t(x))\ ,
\end{equation}
and this $F(S_t(x))$ decreases exponentially fast to its minimum value.

There is useful direct route from \eqref{act6} to \eqref{act9}. 
If $\gamma(t)$ is given by the gradient flow of $F$ through $\gamma(t) = S_t(x)$, then
\begin{equation}\label{act10}
\frac{{\rm d}}{{\rm d}t} F(\gamma(t)) = -\|{\rm grad}_g F(\gamma(t))\|^2_{g(\gamma(t))} \ .
\end{equation}
Define the {\em energy function} $E$ associated to $F$ by
\begin{equation}\label{act10b}
E(x) = \|{\rm grad}_g F(x)\|^2_{g(x)}\ .
\end{equation}
Then (\ref{act6}) applied with $\gamma^t(s) = S_{s+t}(x)$ 
together with the semigroup property yields
\begin{equation*}
\frac{{\rm d}}{{\rm d}t} E(S_t(x)) \leq -2\lambda E(S_t(x))\ .
\end{equation*}
Hence (\ref{act6}) leads directly to 
\begin{equation}\label{act12A}
 E(S_t(x)) \leq e^{-2\lambda t}E(x)\ . 
\end{equation}
Then, as in the original work of Bakry and Emery, \eqref{act12A} together with \eqref{act10} and \eqref{act10b} yield
\begin{equation*}
 F(x) - F(S_T(x)) = \int_0^T E(S_t(x)){\rm d}t \leq \int_0^T e^{-2\lambda t}E(x){\rm d}t \ .
\end{equation*}
Then taking $T\to\infty$ yields
\begin{equation}\label{act13}
F(x) \leq  \frac{1}{2\lambda}E(x)\ . 
\end{equation}
In our setting, when $F$ is a relative entropy function, (\ref{act13}) will be a 
{\em generalized  logarithmic Sobolev inequality}. 
Form here one readily deduces
\begin{equation}\label{act12}
F(S_t(x)) \leq e^{-2\lambda t}F(x) \ . 
\end{equation}

Thus, to prove geodesic convexity of $F$, and what is more directly significant for our applications, the inequalities (\ref{act8}) and (\ref{act12}), it suffices to prove (\ref{act6}), which is a
{\em first derivative estimate}.  In contrast, a  proof of  $\lambda$-convexity of $F$ by direct calculation of the Hessian of $F$ followed by an estimate of its least eigenvalue requires the computation of second covariant derivatives which can be difficult depending on the metric. Otto and Westdickenberg developed their approach having in mind applications to an infinite-dimensional setting (such as that of \cite{O01}) in which second covariant derivatives are difficult to work with rigorously.  It turns out
 that their approach is also quite fruitful in our finite-dimensional setting.

\subsection{Connection with the Brenier-Benamou Formula}

The Brenier-Benamou formula \cite{BB00}  for the $2$-Kantorovich distance between two probability densities $\rho_0$ and $\rho_1$  on $\R^n$ is:
\begin{equation*}\begin{aligned}
	 &W_2(\rho_0, \rho_1)^2\\ & = 
 \inf_{} \bigg\{ \int_0^1 \int  |\nabla \psi_t(x)|^2  \dd \rho_t(x) \dd t \ : \
 		  \partial_t \rho_t + \dive(\rho_t \nabla \psi_t) = 0 \ , 
		  \rho_t|_{t=0,1} = \rho_{0,1} \bigg\}
		  \\& = \inf_{} \bigg\{ \int_0^1\int \frac{|{\bf p}_t(x)|^2}{\rho_t(x)}  \dd x \dd t \ : \
 		  \partial_t \rho_t + \dive {\bf p}_t = 0 \ , 
		  \rho_t|_{t=0,1} = \rho_{0,1} \bigg\} \ .
\end{aligned}\end{equation*}	

The second expression for $W_2(\rho_0, \rho_1)^2$ is the Brenier-Benamou formula. The first expression is used in their work to make a connection with the Wasserstein distance. The equality of the two formulas is a consequence of a minimizing property of gradients:  Among all time dependent vector fields ${\bf p}_t(x)$ such that
\begin{equation}\label{gradmin1}
 \partial_t \rho_t + \dive(\rho_t {\bf p}_t) = 0\ ,
\end{equation}
the one that minimizes ${\displaystyle \int  |\nabla \psi_t(x)|^2  \dd \rho_t(x)}$ is a gradient; i.e., it is of the form ${\bf p}_t = \nabla \psi_t$.   To prove this, one simply needs to show that the set of  ${\bf p}_t$ satisfying \eqref{gradmin1} is non-empty, and then by convexity there is a unique minimizer. Since the difference between any two such ${\bf p}_t$  has zero divergence, a simple integration by parts in the Euler-Lagrange equation provides the proof.
Therefore, in the second minimization formula, in searching for minimizers, ${\bf p}_t$, we may restrict our attention to those of the from ${\bf p}_t = \rho \nabla \psi$. 

It is useful to compare this with our metrics on $\Dens$ that are determined by a QMS generator $\cL$ satisfying the detailed balance condition. 
In our case, for a smooth curve $\rho(t)$  in $\Dens$, we have from Theorem~\ref{ergX2} that
$$g_\rho(\dot\rho,\dot\rho) = \langle \dot \rho, \cK_\rho^{-1}\dot\rho\rangle_\fH   = 
\langle (\cK_\rho^{-1}\dot \rho), \cK_\rho (\cK_\rho^{-1}\dot\rho)\rangle_\fH \ .$$
That is, writing
$\dot \rho(t) =: \cK_{\rho(t)}B(t)$,  $g_\rho(\dot\rho,\dot\rho) = \langle B, \cK_\rho B\rangle$.

Therefore
\begin{align*}
&d_g^2(\rho_0,\rho_1)  =\\ &\inf_{} \bigg\{ \int_0^1 \ip{B_t, \cK_{\rho_t}  B_t} \dd t \ : \
 		  \dot\rho_t = \cK_{\rho_t}B_t\ , 
		  \rho_t|_{t=0} = \rho_0\ ,
		  \rho_t|_{t=1} = \rho_1 \bigg\}\ .
\end{align*}
Recall that $\cK_\rho(B) = -\dive({\mathbb M}_\rho \nabla B)$ where ${\mathbb M}_\rho$ is defined in \eqref{Mudef}.
Making the change of variables,
\begin{equation}\label{gradform}
\widehat{{\bf A}}_t := \mathbb{M}_{\rho(t)} \nabla B_t\ ,
\end{equation}
we have from $\dot\rho_t = \cK_{\rho_t}B_t$ and the definition \eqref{grhodefag} of the metric $g_\rho$,
$$
g_\rho(\dot\rho,\dot\rho) = \langle B_t, \cK_\rho(B_t)\rangle  = \langle  \widehat{{\bf A}}_t ,\mathbb{M}_{\rho(t)}^{-1}\widehat{{\bf A}}_t\rangle \ .
$$

In complete analogy with the  classical case, if we minimize $\langle {\bf A}_t, \mathbb{M}_{\rho(t)}{\bf A}_t\rangle$ over all ${\bf A}_t$ satisfying  $ \dot\rho_t +\dive {\bf A}_t = 0$,  the minimizer ${\bf A}_t$ has the form specified in \eqref{gradform}; that is, it is the gradient of some $B$ as shown in \cite{CM17}. The proof is essentially the same as in the classical case, if anything, it is even simpler in our finite dimensional setting:

\begin{lemma}\label{minpropgr} 
$$g_\rho(\dot\rho,\dot\rho)  =  \min\left\{\ \langle  {\bf A}_t ,\mathbb{M}_{\rho(t)}^{-1}{\bf A}_t\rangle  \ :\  \dot\rho(t) +\dive({\bf A}_t) =0 \ \right\} .$$
\end{lemma}

\begin{proof} 
Let ${\bf A}_t$ be such that $\dot\rho = -\dive({\bf A}_t)$, and let ${\bf X}_t := {\bf A}_t -\widehat{{\bf A}}_t$. Since  $\dot\rho(t) = -\dive(\widehat{{\bf A}}_t)$, $\dive({\bf X}_t )= 0$.
Then
\begin{eqnarray*}
\langle  {\bf A}_t ,\mathbb{M}_{\rho(t)}^{-1}{\bf A}_t\rangle  &=&  
\langle  \widehat{{\bf A}}_t ,\mathbb{M}_{\rho(t)}^{-1}\widehat{{\bf A}}_t\rangle +  \langle  {\bf X}_t ,\mathbb{M}_{\rho(t)}^{-1}{\bf X}_t\rangle\\
&+&2\Re \langle  {\bf X}_t ,\mathbb{M}_{\rho(t)}^{-1}{\widehat{\bf A}}_t\rangle\ .
\end{eqnarray*}
However,  by \eqref{gradform},
$$
\langle  {\bf X}_t ,\mathbb{M}_{\rho(t)}^{-1}{\widehat{\bf A}}_t\rangle = \langle {\bf X}_t, \nabla B\rangle = \langle \dive({\bf X}_t),B\rangle = 0\ ,
$$
and $\langle  {\bf X}_t ,\mathbb{M}_{\rho(t)}^{-1}{\bf X}_t\rangle \geq 0$ with equality if and only if ${\bf A}_t = \widehat{\bf A}_t$.
\end{proof}

Therefore we obtain the following variational formula for our metric:

\begin{align*}
&d_g^2(\rho_0,\rho_1)  =\nonumber\\ &\inf_{} \bigg\{ \int_0^1 \ip{ {\bf A} _t, \mathbb{M}_{\rho(t)}^{-1}{\bf A}_t} \dd t \ : \
 		  \dot \rho_t +\dive {\bf A}_t = 0\ , 
		  \rho_t|_{t=0} = \rho_0\ ,
		  \rho_t|_{t=1} = \rho_1 \bigg\}
\end{align*}
This is the direct analog of the Brenier-Benamou formula,
it gives us the justification for referring to the metrics on $\Dens$ that we are using as {non-commutative mass transport metrics}. More important, just as in the classical setting, the 
``relaxed'' form of the constraint $\dot \rho_t +\dive {\bf A}_t = 0$ is very helpful in applications.

\subsection{Geodesic convexity using intertwining relations}

In  this section, let $\sigma\in \Dens$, and let  $\{\cP_t\}_{t\geq 0} = \{e^{t\cL}\}_{t\geq 0}$ be an ergodic QMS that satisfies the 
detailed balance condition with respect to $\sigma$. 
Let $\cL$ be given in the standard form specified in Theorem~\ref{thm:structure} so that the data specifying 
 $\cL$ are the sets $\{V_j\}_{j\in \cJ}$ and  $\{\omega_j\}_{j\in \cJ}$.  
 Let $\nabla: \fH \to \fH_{\cJ}$ and $\dive: \fH_{\cJ} \to \fH$ be the associated non-commutative gradient and divergence (as opposed to the associated Riemannian gradient and divergence).

Let $\rho:[0,1]\to \Dens_+$ be a smooth path in $\Dens_+$, and define the one-parameter family of paths, $\rho^t(s)$, $(s,t) \in [0,1]\times[0,\infty)$ by 
\begin{equation*}
\rho^t(s)  = \cP_t^\dagger \rho(s)\ .
\end{equation*}
By what has been explained at the beginning of this lecture, if we can prove that
\begin{equation}\label{act23}
\frac{{\rm d}}{{\rm d}t} \bigg|_{0+} 
\left( \left\| \frac{{\rm d}}{{\rm d}s}\rho^t(s)\right\|^2_{g(\rho^t(s))}\right)\leq -2\lambda 
 \left\| \frac{{\rm d}}{{\rm d}s}\rho^0(s)\right\|^2_{g(\rho^0(s))}
\end{equation}
for all smooth $\rho:[0,1]\to \cM$ and all $s\in (0,1)$, we will have 
proved the geodesic convexity of the relative entropy functional, and 
consequently, we shall have proved 
\begin{equation*}
D(\cP_t^\dagger \rho || \sigma) \leq e^{-2\lambda t}D(\rho || \sigma)\ .
\end{equation*}

Note that \eqref{act23} is equivalent to proving
\begin{equation}\label{act23B} 
 \left\| \frac{{\rm d}}{{\rm d}s}\rho^t(s)\right\|^2_{g(\rho^t(s))} \leq e^{-2\lambda t}
 \left\| \frac{{\rm d}}{{\rm d}s}\rho^0(s)\right\|^2_{g(\rho^0(s))}
\end{equation}

By Lemma~\ref{minpropgr}, for any ${\bf C}_s$ such that 
$$
\frac{{\rm d}}{{\rm d}s}\rho^t(s) + {\bf C}_s = 0\ ,
$$
$$
 \left\| \frac{{\rm d}}{{\rm d}s}\rho^t(s)\right\|^2_{g(\rho^t(s))} \leq \langle {\bf C}_s, \mathbb{M}_{\rho^t(s)}^{-1} {\bf C}_s\rangle\ .
$$
On the other hand, if we define $\widehat{{\bf A}}_s := \mathbb{M}_{\rho(s)} \nabla B_s$ where $B_s := \cK_\rho(s)^{-1}\dot\rho(s)$, then again by Lemma~\ref{minpropgr},
$$
\left\| \frac{{\rm d}}{{\rm d}s}\rho^0(s)\right\|^2_{g(\rho^0(s))} = \langle  \widehat{{\bf A}}_s,\mathbb{M}_{\rho(s)}^{-1}  \widehat{{\bf A}}_s\rangle \ .
$$
There is a natural choice for ${\bf C}_s$:   Since $\frac{{\rm d}}{{\rm d}s}\rho^0(s) = -\dive ( \widehat{{\bf A}}_s)$, applying $\cP_t^\dagger$ to both sides yields
$$
\frac{{\rm d}}{{\rm d}s}\rho^t(s) =  \cP_t^\dagger \left( \dive ( \widehat{{\bf A}}_s)\right)\ .
$$
We now construct an ``intertwining operator'' $\cQ_t$ on ${\mathfrak G}_{\cJ}$  such that
$$
\cP_t^\dagger \circ \dive = \dive \circ \cQ_t^\dagger\ .
$$
With this in hand, we may take ${\bf C}_s := \dive\left(\cQ_t^\dagger \widehat{{\bf A}}_s\right)$.  Then, whenever we can prove that
\begin{equation}\label{act23BX} 
\langle  \cQ_t^\dagger\widehat{\bf A}_s, \mathbb{M}_{\cP_t^\dagger \rho(s)}^{-1} \cQ_t^\dagger \widehat{\bf A}_s\rangle \leq e^{-2\lambda t}
\langle  \widehat{{\bf A}}_s,\mathbb{M}_{\rho(s)}^{-1}  \widehat{{\bf A}}_s\rangle\ ,
\end{equation}
we shall have proved \eqref{act23B}.  We shall be able to do this in a number of interesting cases. First, we construct the map $\cQ_t$. 

\begin{lemma}\label{Qconslm} For $t>0$, define the linear transformation $\cQ_t:{\mathfrak G}_{\cJ} \to  {\mathfrak G}_{\cJ}$ by
\begin{equation}\label{Qtdef}
\cQ_t := \nabla \cP_t \circ \cL_0^+ \circ \dive
\end{equation}
where $\cL_0^+$ is the generalized inverse of the map $\cL_0$ defined in \eqref{cL0def}.
Then for all traceless $A$,
\begin{equation}\label{intertwine1}
\nabla \cP_t  A= \cQ_t\nabla A \ ,
\end{equation}
and consequently
\begin{equation*}
 \cP_t^\dagger  \circ \dive = \dive \circ \cQ_t^\dagger \ 
\end{equation*}
as operators on ${\mathfrak H}_{\cJ}$. Furthermore, for all $s,t>0$,
\begin{equation}\label{intertwine3}
\cQ_s \cQ_t = \cQ_{s+t}\ .
\end{equation}
\end{lemma}

\begin{proof} A simple computation verifies \eqref{intertwine1}.  Then for any ${\bf C}\in {\mathfrak H}_{\cJ}$,
$$
\langle {\bf C}, \nabla \cP_t  A\rangle_{\fH_{\cJ}} =  \langle {\bf C}, \cQ_t\nabla A\rangle_{\fH_{\cJ}} = \langle \cQ_t^\dagger {\bf C}, \nabla A\rangle_{\fH_{\cJ}}\ .
$$
Integration by parts yields
$$
\langle \cP_t^\dagger(\dive {\bf C}),A\rangle_{\fH} = \langle \dive(\cQ_t({\bf C}), A\rangle 
$$
for all traceless $A$, and then since divergences are traceless and $\cP_t^\dagger$ is trace preserving, the restriction on the trace of $A$ is superfluous. 
Finally, by the definition
$$
\cQ_s \cQ_t  = \nabla \cP_s \circ \cL_0^+\cL_0  \cP_t \circ \cL_0^+ \circ \dive =  \nabla \cP_{s+t} \circ \cL_0^+ \circ \dive = \cQ_{s+t}\ 
$$
since $\cP_t$ preserves the kernel of $\cL_0$, and this proves \eqref{intertwine3}.
\end{proof}

\begin{definition} The inequality 
\begin{equation}\label{act23BX2} 
\langle  \cQ_t^\dagger{\bf X} , \mathbb{M}_{\cP_t^\dagger \rho(s)}^{-1} \cQ_t^\dagger{\bf X}\rangle \leq e^{-2\lambda t}
\langle  {{\bf X}},\mathbb{M}_{\rho(s)}^{-1}  {{\bf X}}\rangle\ ,
\end{equation} when it holds for some $\lambda \geq 0$, and all ${\bf X} \in \fH$,  is called the {\em action dissipation inequality}.
\end{definition}

Evidently when \eqref{act23BX2} is valid for general ${\bf X}$, the specific case \eqref{act23BX} is valid.
We now present a simple sufficient condition for \eqref{act23BX2}  that can 
 be able  verified in a number of interesting examples, as was done in \cite{CM17}. 

In some case, $\cQ_t$ takes a particularly simple form such that \eqref{act23BX} is evidently satisfied. 

\begin{theorem}\label{odelm} Suppose that for some numbers $a_j$, $j\in \cJ$, 
\begin{align}\label{intertw2A}
[\partial_j,\cL]  = -a_j\partial_j \ 
\end{align}
for each $j\in \cJ$.  Then with $\cQ_t:{\mathfrak G}_{\cJ} \to  {\mathfrak G}_{\cJ}$   defined by \eqref{Qtdef},
\begin{equation}\label{intertw2B}
 \cQ_t(A_1,\dots,A_{|\cJ|}) = (e^{-ta_1} \cP_t A_1,\dots,e^{-ta_{|\cJ|}}\cP_t A_{|\cJ|})\ .
\end{equation}
Then with $\lambda := \min\{\ a_j\ :\ j\in \cJ\ \}$, the action dissipation inequality  \eqref{act23BX2}  is satisfied.
\end{theorem}

\begin{proof} Let $A\in \fH$, $\tr[A] =0$,  and define $A_j(t) = \partial_j \cP_t A$. Then $A_j(0) =\partial_j A_j$ and 
$$\frac{{\rm d}}{{\rm d}t}A_j(t) = \partial_j \cL \cP_t A = 
\cL \partial_j \cP_t A  - a_j \partial_j \cP_t A = [\cL - a_j \one]A_j(t)\ .$$
It follows that $t \mapsto e^{ta_j}A_j(t)$ is the unique solution of
${\displaystyle \frac{{\rm d}}{{\rm d}t}X(t) = \cL X(t)}$   with 
$X(0) =\partial_j A$, which is of course $\cP_t \partial_j A$. Therefore,
$\partial_j\cP_t A  =   e^{-ta_j} \cP_t \partial_j A $, and this proves that 
\begin{equation}\label{intertwine6}
\nabla \cP_tA = (e^{-ta_1} \cP_t \partial_1 A,\dots,e^{-ta_{|\cJ|}}\cP_t \partial_{|\cJ|} A)\ .
\end{equation}
Then since $\tr[A] =0$, $\cL_0^+ \circ\dive (\nabla A) = A$, and then \eqref{intertwine6} yields
$$
\cQ_t\nabla A = \nabla \cP_t  \cL_0^+ \circ\dive (\nabla A) = (e^{-ta_1} \cP_t \partial_1 A,\dots,e^{-ta_{|\cJ|}}\cP_t \partial_{|\cJ|} A)\ ,
$$
and this proves \eqref{intertw2B}.  This in turn yields
$$
\langle  \cQ_t^\dagger\widehat{\bf A}_s, \mathbb{M}_{\cP_t^\dagger \rho(s)}^{-1} \cQ_t^\dagger \widehat{\bf A}_s\rangle =e^{-2\lambda t} 
\langle  \cP_t^\dagger\widehat{\bf A}_s, \mathbb{M}_{\cP_t^\dagger \rho(s)}^{-1} \cP_t^\dagger\widehat{\bf A}_s\rangle \ ,
$$
and then by Theorem~\ref{L3M}
$$
\langle  \cP_t^\dagger\widehat{\bf A}_s, \mathbb{M}_{\cP_t^\dagger \rho(s)}^{-1} \cP_t^\dagger \widehat{\bf A}_s\rangle \leq \langle  \widehat{\bf A}_s, \mathbb{M}_{ \rho(s)}^{-1}\widehat{\bf A}_s \rangle
$$
Combining the last two inequalities yields \eqref{act23BX2} .
\end{proof}

 It was shown in  \cite{CM17} that for both the Fermi and Bose Quantum Mehler Semgroups one has that 
\begin{equation}\label{goodcase}
\cQ_t {\bf A} =  
(e^{-\lambda t} \cP_t A_1,\dots, e^{-\lambda t}\cP_t A_{|\cJ|})\ .
\end{equation}
for some $\lambda>0$ depending on the ``temperature'' parameter in the semigroup, and on whether the semigroup is the Fermi or Bose version. See \cite{CM17} for details.  Here we simply note that whenever \eqref{goodcase} is satisfied, 
$$
\langle  \cQ_t^\dagger\widehat{\bf A}_s, \mathbb{M}_{\cP_t^\dagger \rho(s)}^{-1} \cQ_t^\dagger \widehat{\bf A}_s\rangle =e^{-2\lambda t} 
\langle  \cP_t^\dagger\widehat{\bf A}_s, \mathbb{M}_{\cP_t^\dagger \rho(s)}^{-1} \cP_t^\dagger\widehat{\bf A}_s\rangle \ ,
$$
and then by Theorem~\ref{L3M}
$$
\langle  \cP_t^\dagger\widehat{\bf A}_s, \mathbb{M}_{\cP_t^\dagger \rho(s)}^{-1} \cP_t^\dagger \widehat{\bf A}_s\rangle \leq \langle  \widehat{\bf A}_s, \mathbb{M}_{ \rho(s)}^{-1}\widehat{\bf A}_s \rangle
$$
Therefore,
whenever \eqref{goodcase} is valid,  so is \eqref{act23BX} and hence \eqref{act23B}. 
Altogether we have proved:

\begin{theorem}\label{sumthm} 
Let $\sigma\in \Dens_+$, and let $\cP_t = e^{t\cL}$ be an ergodic QMS that satisfies the detailed balance condition with respect to $\sigma$.
Let $\nabla$ and $\dive$ denote the associated non-commutative gradient and divergence. 
Suppose that for some $\lambda> 0$, the action dissipation inequality \eqref{act23BX2} is satisfied.

Then the relative entropy with respect to $\sigma$ is geodesicaly $\lambda$-convex  for the Riemannian metric $g_{\rho}$ determined by $\cL$.
Moreover, the exponential convergence estimate
\begin{align*}
	D(\cP_t^\dagger \rho || \sigma) \leq e^{-2\lambda t}D(\rho || \sigma)
\end{align*}
holds, as well as the generalized logarithmic Sobolev inequality
\begin{align}\label{eq:gen-LSI}
	D(\rho || \sigma) \leq  -\frac{1}{2\lambda} \tr\big[\cL^\dagger(\rho)\big(\log \rho - \log \sigma\big)\big] \ .
\end{align}
\end{theorem}

In \cite{CM17}, Theorem~\ref{sumthm} was shown to apply
for both the Fermi and Bose Ornstein-Uhlenbeck semigroups $\cP_t$, with the action dissipation inequality being proved using Theorem~\ref{odelm}.
Actually, in these cases,  there is a 
Mehler type formula for 
$\cP_t$ from which the intertwining can be readily checked, as shown in \cite{CM17}. In the Fermi case, 
this can be found in formulas (4.1) and (4.2) of \cite{CL93}, and the formula in the Bose case is a simple adaptation of this.

\subsection{Dual from of the action dissipation inequality}

There is a dual form of the action dissipation inequality \eqref{act23BX2},
$$
\langle  \cQ_t^\dagger{\bf X} , \mathbb{M}_{\cP_t^\dagger \rho(s)}^{-1} \cQ_t^\dagger{\bf X}\rangle \leq e^{-2\lambda t}
\langle  {{\bf X}},\mathbb{M}_{\rho(s)}^{-1}  {{\bf X}}\rangle\ ,
$$
that is particularly useful in applications. 

Suppose \eqref{act23BX2} is valid. Then for arbitrary ${\bf Y}\in M_n(\C)$, 
\begin{eqnarray*}
\langle \nabla Y, \cQ_t^\dagger{\bf X}\rangle - \frac12  \langle  \cQ_t^\dagger{\bf X} , \mathbb{M}_{\cP_t^\dagger \rho(s)}^{-1} \cQ_t^\dagger{\bf X}\rangle 
&\geq& \langle \nabla Y, \cQ_t^\dagger{\bf X}\rangle - \frac12  e^{-2\lambda t} \langle  {{\bf X}},\mathbb{M}_{\rho(s)}^{-1}  {{\bf X}}\rangle\\
&=& \langle \cQ_t\nabla Y, {\bf X}\rangle - \frac12  e^{-2\lambda t} \langle  {{\bf X}},\mathbb{M}_{\rho(s)}^{-1}  {{\bf X}}\rangle\\
&=& \langle \nabla  \cP_tY, {\bf X}\rangle - \frac12  e^{-2\lambda t} \langle  {{\bf X}},\mathbb{M}_{\rho(s)}^{-1}  {{\bf X}}\rangle\ .
\end{eqnarray*}
Taking the supremum over ${\bf X}$, so that we are taking the Legendre transforms of two quadratic functions we obtain
\begin{equation}\label{act23BX3}
\langle \nabla Y, \mathbb{M}_{\cP_t^\dagger \rho(s)} \nabla Y\rangle \geq e^{2\lambda t} \langle \nabla \cP_t Y, \mathbb{M}_{\rho} \nabla \cP_t Y\rangle\ .
\end{equation}

Since the Legendre transform is involutive, a standard argument shows that \eqref{act23BX3} is actually equivalent to the  action dissipation inequality \eqref{act23BX2}.

The inequality \eqref{act23BX3} is called the {\em gradient estimate} by Wirth and Zhang \cite{WZ23}. Actually it is the ``dimension free''  case $GE(\lambda,\infty)$ of a family of inequalities $GE(\lambda,N)$ depending on ``curvature'' and ``dimension''.     They show, using a calculation of the Hessian of the relative entropy from \cite{CM20} that  \eqref{act23BX3} holds if and only if the relative entropy is $\lambda$ geodesically convex.  While the line of argument presented here may suggest that 
 that the entropy dissipation inequality is  only a sufficient condition for  $\lambda$-geodesic convexity of the relative entropy,  it is actually necessary as well.  This was proved in \cite{CM20,DR20}.   The  equivalence by duality of  \eqref{act23BX}   and  \eqref{act23BX3} somewhat simplifies the proof that either condition is equivalent to geodesic convexity of the relative entropy. However, the main point to make as the end of these lectures approaches is that there is much additional information in the papers \cite{CM20,DR20,WZ23}.

Wirth and Zhang also show using concavity of operator means that whenever the semigroup $\cP_t$ satisfies the intertwining relation \eqref{goodcase}, then  the gradient estimate  \eqref{act23BX3} is satisfied. We gave a simple proof using the inequalities considered here: Assuming \eqref{goodcase},
$$
e^{2\lambda t} \langle \nabla \cP_t Y, \mathbb{M}_{\rho} \nabla \cP_t Y\rangle = \langle \cP_t \nabla Y, \mathbb{M}_{\rho}  \cP_t  \nabla Y\rangle
$$
and then by Theorem~\ref{L1M}, Uhlmann's  monotonicity version of the Lieb Concavity Theorem,
$$\langle \cP_t \nabla Y, \mathbb{M}_{\rho}  \cP_t  \nabla Y\rangle \leq  \langle \nabla Y, \mathbb{M}_{\cP_t^\dagger \rho(s)} \nabla Y\rangle \ .
$$

\subsection{Another route to the action dissipation inequality}

Theorem~\ref{odelm} is not the only route to proving the action dissipation inequality. We close with a simple example due to Wirth and Zhang \cite{WZ23}.  The {\em completely depolarizing channel} is the QMS on $M_n(\C)$ given by
\begin{equation}\label{cdep1}
\cL(X) = \frac1n\tr[X]\one -X\ .
\end{equation}
Since $\Phi(X) := \frac1n\tr[X]\one$ defined a completely positive unital and trace preserving map, $\cP_t = e^{t\cL}$ is a QMS satisfying the detailed balance condition with respect to $\tau$, the normalized trace. 

Using \eqref{cdep1} in the right side of \eqref{eq:gen-LSI} with $\sigma= \tau$, we obtain
$$
-\tr\big[\cL^\dagger(\rho)\big(\log \rho - \log \tau\big)\big] = -\tr[ (\tau -\rho) \big(\log \rho - \log \tau\big)\big] = D(\rho||\tau) + D(\tau||\rho)\ ,
$$
as noted in \cite{MHFW16}. Since $D(\tau||\rho) \geq 0$, we have the generalized logarithmic Sobolev inequality
\begin{equation}\label{CDCex1}
D(\rho||\tau)  \leq -\tr\big[\cL^\dagger(\rho)\big(\log \rho - \log \tau\big)\big] \ .
\end{equation}
This inequality turns out to be the  best possible inequality that holds {\em independent of the dimension}: 
In high dimension there exist $\rho \neq \tau$ such that $D(\tau||\rho)$ is an arbitrarily small multiple of $D(\rho||\tau)$;  \cite{MHFW16}. 

It is easy to write the generator explicitly in terms of a set $\{V_1,\dots,V_{2n-1}\}$ of operators defining that various partial derivatives, but we shall not need this. Also now that since for $\tau$, the modular operator is the identity, all of the Bohr frequencies $\omega$ are zero. Hence $\mathbb{M}_{\rho} $ simplifies to 
$$
\mathbb{M}_{\rho}{\bf  A} = \int_0^1 \rho^{1-s}{\bf A} \rho^{s}{\rm d}s\ .
$$
It is also easy to write down a formula for $\cP_t$ in this simple case:
$$
\cP_t A = e^{-t}A + (1-e^{-t})\frac1n \tr[A]\one\ .
$$

Since $\nabla \one =0$, and since each component of $\nabla A$ is traceless,
\begin{equation}\label{interwz}
\nabla \cP_tA = e^{-t}\nabla A = \cP_t \nabla A\ .
\end{equation}
Equivalntly, for each $j$, 
$$
[\partial_j,\cL] = 0
$$
so that \eqref{intertw2B} is satisfied, but for $a_j= 0$. Hence  Theorem~\ref{odelm}  only tells us that the action dissipation inequality  \eqref{act23BX3} is satisfied
 for $\lambda = 0$, which is useless.     

However, 
as Wirth and Zhang showed, 
one can still establish a non-trivial form of  \eqref{act23BX3}: First observe that since $\Phi$, and hence each $\cP_t$ is self adjoint, so that $\cP_t^\dagger = \cP_t$,
$$
\langle \nabla Y, \mathbb{M}_{\cP_t^\dagger \rho(s)} \nabla Y\rangle = \int_0^1 \tr[(\cP_t)^s\rho (\nabla Y)^*  (\cP_t)^{1-s}\rho(\nabla Y)]{\rm d}s
$$
Then by the Lieb Concavity Theorem, for each $s$
$$
\tr[(\cP_t\rho )^s (\nabla Y)^*  (\cP_t\rho)^{1-s} (\nabla Y)] \geq e^{-t} \tr[\rho^s (\nabla Y)^*\rho^{1-s}(\nabla Y)] \ , 
$$
where we have simply discarded the positive term multiplying $(1-e^{-t})$. Therefore,
$$\langle \nabla Y, \mathbb{M}_{\cP_t^\dagger \rho(s)} \nabla Y\rangle \geq e^{-t} \langle \nabla Y, \mathbb{M}_\rho \nabla Y\rangle\ .
$$
However, by \eqref{interwz},
$$
e^{-t} \langle \nabla Y, \mathbb{M}_\rho \nabla Y\rangle =  e^{t} \langle e^{-t}\nabla Y, \mathbb{M}_\rho e^{-t}\nabla Y\rangle = e^t \langle \nabla \cP_t Y, \mathbb{M}_\rho  \nabla \cP_t Y\rangle\ .
$$
Therefore, \eqref{act23BX3} is satisfied with $\lambda =1/2$.  Then by Theorem~\ref{sumthm}, we recover \eqref{CDCex1}. For the dimension dependent best constant in the entropy production inequality for the depolarizing channel, see \cite{MHFW16}.  The simple treatment 
here does not yield the best constant in low dimensions, but using the curvature-dimension conditions investigated in \cite{WZ20,WZ23}, one can do better. It remain to be seen if the method presented here can be adapted to treat other problems such as those considered in \cite{BDP23}.

\bigskip
\noindent{\bf Acknowledgement:} Work partially supported by U.S.
National Science Foundation grant  DMS 2055282. I thank Haonan Zhang for a careful reading of the first draft of this paper, and for many helpful comments.